\documentclass[a4paper,12pt]{article}

\usepackage{amsfonts}
\usepackage{amsmath}
\usepackage{amsthm}
\usepackage{mathrsfs}
\usepackage[all]{xy}
\usepackage{verbatim}
\usepackage{rotating}
\usepackage{graphicx}
 \usepackage{longtable}
\usepackage{multirow}
\usepackage[ps2pdf]{hyperref}

\hypersetup{pdftitle = {Permutation orbifolds of heterotic Gepner models}}
\hypersetup{pdfauthor = {Michele Maio, Bert Schellekens}}
\hypersetup{pdfsubject = {CFT-String Theory}}
\hypersetup{pdfkeywords = {Orbifolds, Extensions, Fixed Point Resolution, N=2 minimal models}}

\numberwithin{equation}{section}

\newtheorem{theorem}{Theorem}[section]

\begin{document}
\title{Permutation orbifolds of heterotic Gepner models}
\author{M. Maio$^{1}$ and A.N. Schellekens$^{1,2,3}$\\~\\~\\
\\
$^1$Nikhef Theory Group, Amsterdam, The Netherlands
\\
~\\
$^2$IMAPP, Radboud Universiteit Nijmegen, The Netherlands
\\
~\\
$^3$Instituto de F\'\i sica Fundamental, CSIC, Madrid, Spain}

\maketitle

\begin{abstract}
We study orbifolds by permutations of two identical N=2 minimal models within the Gepner construction of four dimensional heterotic strings. 
This is done using the new N=2 supersymmetric permutation orbifold building blocks we  have recently developed. We compare
our results with the old method of modding out the full string partition function. The overlap between these two approaches is
surprisingly small, but whenever a comparison can be made we find complete agreement. The use of permutation building blocks
allows us to use the complete arsenal of simple current techniques that is available for standard Gepner models, vastly extending what
could previously be done for permutation orbifolds. In particular, we consider $(0,2)$ models, breaking of $SO(10)$ to subgroups,
weight-lifting for the minimal models and B-L lifting. Some previously observed phenomena, for example concerning family number 
quantization, extend to this new class as well, and in the lifted models three family models occur with abundance comparable to two or four.
%
%
\end{abstract}
\vskip -19cm
\hbox{ \hskip 12cm NIKHEF/2011-004  \hfil}

\clearpage
\tableofcontents

\section{Introduction}

Heterotic string theory \cite{Gross:1984dd} is the oldest approach towards the construction of the standard model in string theory. It
owes its success to the fact that the gross features of the standard model appear to come out nearly automatically:
families of chiral fermions in representations that are structured as in $SO(10)$-based GUT models.
Presumably we know only a very small piece of this landscape of perturbative heterotic strings.
Several computational methods are available, based either on space-time geometry or world-sheet conformal field theory.
The latter category can be further subdivided into free and interacting field theory methods. But all these methods have 
obvious but inestimable limitations, making it difficult to determine how much of the heterotic landscape has really been explored so far.

In the case of the conformal field theory (CFT) methods \cite{Belavin:1984vu}, these limitations are due to the constraints imposed by world-sheet supersymmetry and
modular invariance, for which we only know some classes of special solutions.
A general heterotic CFT consists of a right-moving sector that has $N=2$ world-sheet supersymmetry and a non-supersymmetric
left-moving sector. Most existing work has been limited either to free CFTs (bosons, fermions or orbifolds) for these two sectors, 
or to interacting CFTs where the bosonic sector is essentially a copy of the fermionic one. Furthermore the interacting CFTs
themselves have mostly been limited to tensor products of $N=2$ minimal models \cite{Gepner:1987qi,Gepner:1987vz}.

Already in the late eighties of last century ideas were implemented to reduce some of these limitations of interacting CFTs.
Instead of minimal models, Kazama-Suzuki  models were used \cite{Kazama:1988qp}. Another extension was to consider permutation orbifolds
of $N=2$ minimal models \cite{Klemm:1990df,Fuchs:1991vu}. But both of these ideas could only be analyzed in a very limited way themselves. The real power of
interacting CFT construction comes from the use of simple current invariants  
\cite{Schellekens:1989am,Schellekens:1989dq,Intriligator:1989zw,GatoRivera:1991ru,Kreuzer:1993tf},
which greatly enhance the number and scope
of the possible constructions. In particular
the left-right symmetry of the original Gepner models could be broken by
considering asymmetric simple current invariants \cite{Schellekens:1989wx}, allowing for example a breaking of the canonical $E_6$ subgroup to
$SO(10)$, $SU(5)$, Pati-Salam models  or even just the standard model (with some additional factors in the
gauge group). However, precisely this powerful tool is not available at present in either Kazama-Suzuki models or
permutation orbifolds. The original computations were limited to diagonal invariants, where with a combination of a variety
of tricks the spectrum could be obtained. Up to now, all that is available in the literature is a very short list of Hodge numbers
and singlets for $(2,2)$ spectra with $E_6$ gauge groups \cite{Klemm:1990df,Fuchs:1991vu,Font:1989qc,Schellekens:1991sb}.
To use the full power of simple current methods we need to know the exact CFT spectrum
and the fusion rules of the primary fields of the  building blocks. 
The former has never been worked out for Kazama-Suzuki models, and the latter
was not available for permutation orbifolds until last year. Using pioneering work by Borisov, Halpern and Schweigert  \cite{Borisov:1997nc} and an extension of these results to fixed point resolution matrices \cite{Maio:2009kb,Maio:2009cy,Maio:2009tg} we have
constructed the $\mathbb{Z}_2$ permutation orbifolds of $N=2$ minimal models \cite{Maio:2010eu}. These can now be used as building blocks
in heterotic CFT constructions, on equal footing, and in combination with all other building blocks, such as  the minimal
models themselves and free fermions. Furthermore we can now for the first time apply the full simple current machinery in exactly
the same way as for the minimal models. 

Meanwhile, another method was added to this toolbox, allowing us to advance a bit more deeply into the heterotic landscape, and
away from free or symmetric CFTs. This is called ``heterotic weight lifting" \cite{GatoRivera:2009yt}, a replacement of $N=2$ building blocks in the bosonic
sector by isomorphic (in the sense of the modular group) $N=0$ building blocks (more precisely, replacing $N=2$ building blocks together
with the extra $E_8$ factor). This method requires knowledge of the exact CFT spectrum, which indeed we have. A variant of this idea is
the replacement of the $U(1)_{B-L}$ factor ($\times E_8$) by an isomorphic CFT. This has been called ``B-L lifting".

In a series of papers \cite{GatoRivera:2010gv,GatoRivera:2010xn,GatoRivera:2010fi,GatoRivera:2010xnb}, all combinations of the aforementioned tools were applied to Gepner models. The focus of this work was on three important features of the
spectra of these models: do the chiral fermions really form standard model families, do the non-chiral fermions respect
standard model charge quantization, and do three families occur reasonably often. The answers to these questions can be
summarized as ``sometimes, no, and yes". A substantial fraction of the chiral spectra with broken GUT symmetry
is organized in standard model families, but the rest has chiral exotics. In the exact spectrum 
there are almost always vector-like states which violate
the charge quantization rules of the standard model and hence would lead to light fractionally charged particles, unless they are
lifted by moving into moduli space, away from the special RCFT points.  The conclusion on the number of families was a bit more
optimistic. Despite the difficulties of standard Gepner models to yield three families, this problem disappears if one moves away
from symmetric CFTs. Just using asymmetric MIPFs is not sufficient, but if one really makes the left and right CFTs distinct by
using heterotic weight lifting or B-L lifting one obtains a normal, more or less exponential distribution $\propto e^{-\alpha N}$ with
$\alpha$ in the range $.2 - .3$, where $N$  is the number of families. The number three is not strongly suppressed  in these distributions, and occurs for a total
of about $15\%$ of the $N$-family models. 

The purpose of this paper is to put all these ingredients together using permutation orbifolds of $N=2$ minimal models as building 
blocks in combination with minimal models. We want to do this for the following reasons:
\begin{itemize}
\item{Check the consistency of the permutation orbifold CFTs we presented in \cite{Maio:2010eu}. Chiral heterotic spectra are very sensitive to the correctness of conformal weights and ground state dimensions of the CFT, as well as the
correctness of the simple current orbits. This is especially true for weight-lifted spectra, because
they have non-trivial Green-Schwarz anomaly cancellations.}
\item{Compare our results with those of previous work on permutation orbifolds \cite{Klemm:1990df,Fuchs:1991vu}. These results
were obtained using a rather different method, by applying permutations directly to complete heterotic string spectra.}
\item{Check if the aforementioned trends on fractional charged and family number are confirmed also in the class of permutation orbifolds.}
\item{Add a few more items to the growing list of potentially interesting three-family interacting CFT models.}
\end{itemize}

%

The key ingredient of the present work is
is our previous paper \cite{Maio:2010eu}, where we studied permutations, together with extensions in all possible order, and found very interesting novelties. For example, we determined how to construct a supersymmetric permutation of minimal models: in particular, the world-sheet-supersymmetry current in the supersymmetric orbifold turns out to be related to the anti-symmetric representation of the world-sheet-supersymmetry current of the original minimal model. When the symmetric representation is used, instead, one ends up with a conformal field theory, which is isomorphic to the supersymmetric orbifold, but it is not supersymmetric itself. 

In the extended permuted orbifolds so-called exceptional simple currents appear, which originate from off-diagonal representations. Generically, there are many of them, depending on the particular model under consideration, and they do not have fixed points. However, when and only when the ``level'' is equal to $k=2$ mod $4$, four of all these exceptional currents do admit fixed points. As a consequence, in those cases the knowledge of the modular $S$ matrix is plagued by the existence of non-trivial and unknown $S^J$ matrices (one $S^J$ matrix for each exceptional current $J$). The full set of $S^J$ matrices is available for standard $\mathbb{Z}_2$ orbifolds (see \cite{Maio:2009kb,Maio:2009cy,Maio:2009tg}), but not for their (non\-)supersymmetric extensions, due to these four exceptional currents with fixed points. This is known as the fixed point resolution problem 
\cite{Schellekens:1989uf,Schellekens:1990xy,Fuchs:1996dd,Fuchs:1995zr,Schellekens:1999yg,Fuchs:1995tq}.
The knowledge of the full set of the $S^J$ matrices is more important in orientifold models \cite{Pradisi:1996yd,Pradisi:1995pp,Fuchs:2000cm} than in heterotic strings, because in the former case one cannot even compute spectra if these matrices are not known.

In this paper we consider permutations in Gepner models. One starts with Gepner's standard construction where the internal CFT is a product of $N=2$ minimal models. Sometimes there are (at least) two $N=2$ identical factors in the tensor product. When it is the case, we can replace these two factors with their permutation orbifold. Moreover, one also has to impose space-time and world-sheet supersymmetry, which is achieved by suitable simple-current extensions.

The plan of the paper is as follows. 
In section \ref{Section: Permutation orbifold} we start reviewing the permutation orbifold, which is the main idea used in this paper. 
In section \ref{Section: Heterotic Gepner models} we review the standard constructions of Heterotic Gepner models. In section \ref{Section: Permutation orbifold of N=2 minimal models} we review the main ingredients and the most relevant results of $\mathbb{Z}_2$ permutation orbifolds when applied to $N=2$ minimal models. 
In section \ref{Section: Lifts} we describe the heterotic weight lifting and the B-L lifting procedures, which allow us to replace the trivial $E_8$ factor plus either one $N=2$ minimal model or the $U(1)_{B-L}$ with a different CFT, which has identical modular properties, in the bosonic (left) sector. 
In section \ref{Section: Comparison} we compare our results on (2,2) spectra with the known literature. 
In section \ref{Section: Results} we present our phenomenological results concerning the family number distributions,
gauge groups, fractional charges and other relevant data. 
In appendix \ref{Section: Simple current invariants} we derive a few facts about simple current invariants.  Appendix A.2 contains
tables summarizing the main results for the four cases (standard Gepner models and the three kinds of lifts).

\section{Permutation orbifold}
\label{Section: Permutation orbifold}
In this section let us recall a few properties of the generic permutation orbifold \cite{Borisov:1997nc}, restricted to the $\mathbb{Z}_2$ case:
\begin{equation}
\mathcal{A}_{\rm perm}\equiv \mathcal{A}\times \mathcal{A}/\mathbb{Z}_2\,.
\end{equation}
If $c$ is the central charge of $\mathcal{A}$, then the central charge of $\mathcal{A}_{\rm perm}$ is $2c$.
The typical\footnote{See \cite{Maio:2010eu} for exceptions.} weights of the fields are:
\begin{itemize}
\item $h_{(i,\xi)}=2 h_i$
\item $h_{\langle i,j\rangle}=h_i+h_j$
\item $h\widehat{(i,\xi)}=\frac{h_i}{2}+\frac{c}{16}+\frac{\xi}{2}$
\end{itemize}
for diagonal, off-diagonal and twisted representations. 
Sometimes it can happen that the naive ground state has dimension zero: then one must go to its first non-vanishing descendant whose weight is incremented by integers. 

The permutation orbifold characters $X$ are related to the characters $\chi$ of the mother theory $\mathcal{A}$ in the following way \cite{Borisov:1997nc}:
\begin{subequations}
\label{BHS characters}
\begin{eqnarray}
X_{\langle i,j \rangle}(\tau)&=&\chi_{i}(\tau)\cdot\chi_{j}(\tau) \\
X_{(i,\xi)}(\tau)&=&\frac{1}{2}\chi_{i}^2(\tau)+e^{i\pi\xi}\frac{1}{2}\chi_{i}(2\tau) \\
X_{\widehat{(i,\xi)}}(\tau)&=&\frac{1}{2}\chi_{i}(\frac{\tau}{2})+e^{-i\pi\xi}\,T_i^{-\frac{1}{2}}\,\frac{1}{2}\chi_{i}(\frac{\tau+1}{2})
\end{eqnarray}
\end{subequations}
where $T_i^{-\frac{1}{2}}=e^{-i\pi(h_i-\frac{c}{24})}$. 
The character expansion:
\begin{equation}
\chi(\tau)=q^{h_{\chi}-\frac{c}{24}}\,\sum_{n=0}^\infty d_n q^n \qquad \qquad ({\rm with}\,\,q=e^{2i\pi\tau})
\end{equation}
(the $d_n$'s are non-negative integers) implies a similar expansion for the $X$:
\begin{equation}
X(\tau)=q^{h_X-\frac{c}{12}}\,\sum_{n=0}^\infty D_n q^n
\end{equation}
where the $D_n$'s are expressed in terms of the $d_n$'s as follows:
\begin{subequations}
\label{d-D relations}
\begin{eqnarray}
D^{\langle i,j \rangle}_k &=& \sum_{n=0}^k d_n^{(i)}\,d_{k-n}^{(j)} \\
D^{(i,\xi)}_k &=&  \frac{1}{2} \sum_{n=0}^k d_n^{(i)}\,d_{k-n}^{(i)}+
\left\{
\begin{array}{lr}
0 & {\rm if}\,\,k={\rm odd} \\
\frac{1}{2}\,e^{i\pi\xi}\,d_{\frac{k}{2}}^{(i)} & {\rm if}\,\,k={\rm even}
\end{array}
\right. \\
D^{\widehat{(i,\xi)}}_k &=& d_{2k+\xi}^{(i)}
\end{eqnarray}
\end{subequations}

Using these characters, one can compute their modular transformation and find the orbifold $S$ matrix. It was determined in \cite{Borisov:1997nc} and will be referred to as $S^{BHS}$. 
In theories where the permutation orbifold is extended by means of simple currents, one has to face the problem of fixed point resolution, namely one has to determine a set of ``$S^J$ matrices'', one for each simple current $J$ \cite{Schellekens:1990xy}. 
For standard permutation orbifolds, it was shown in \cite{Maio:2009kb} that the simple currents can be only of diagonal type, i.e. $(J,\psi)$ where $J$ is a simple current of the mother theory and $\psi=0$ or $1$ for $\mathbb{Z}_2$ orbifolds. This set of matrices was determined in \cite{Maio:2009tg}. This fixed point resolution plays a crucial r\^ole in determining the fusion rules of the supersymmetric
permutation orbifold, as explained in \cite{Maio:2010eu}. However, we will not need the explicit formulas for the fixed point resolution matrices in
the present paper, and therefore we will not present them here.


\section{Heterotic Gepner models}
\label{Section: Heterotic Gepner models}
In this section we review the construction of four-dimensional heterotic string theory. The starting point is a set of bosons $X^\mu$ ($\mu=0,\dots,3$) for both the right and left movers, a right-moving set of NSR fermions $\psi^\mu$, plus corresponding ghosts, and an internal CFT with central charges $(c_L,c_R)=(22,9)$, that we denote by $\mathscr{C}_{22,9}=\mathscr{C}_{22}\times\mathscr{C}_{9}$. Observe that the right-moving superconformal field theory $(X,\psi)$+ghosts has central charge $c=3$. Equivalently, one can think of it as the conformal field theory of two bosons $X^i$ and their fermionic superpartners $\psi^i$ in light-cone gauge, which form an $SO(2)_1$ abelian algebra, with central charge $c=1$.

The next step is to replace the NSR $SO(2)_1$ fermions by a set of $13$ bosonic fields living in the maximal torus of an $(E_8)_1\times SO(10)_1$ affine Lie algebra. This is the bosonic string map \cite {LLS}, which transforms the fermionic CFT into a bosonic one with same modular properties. The total right-moving CFT has now central charge equal to $c_R=2+9+13=24$, as the left-moving bosonic theory. Hence, all four-dimensional heterotic strings correspond to all compactified bosonic strings with an internal sector:
\begin{equation}
\mathscr{C}_{22,9}\times \left( (E_8)_1\times SO(10)_1 \right)_R \,.
\end{equation}
To summarize:
\begin{eqnarray}
\hbox{Left-moving}&&(X^\mu,{\rm ghost})\times \mathscr{C}_{22}\nonumber\\
\hbox{Right-moving}&&(X^\mu,{\rm ghost})\times \mathscr{C}_{9}\times (E_8)_1\times SO(10)_1\nonumber
\end{eqnarray}
with $\mu=0,\dots,3$. Equivalently, in light-cone gauge one uses $X^i$ instead of $(X^\mu,{\rm ghost})$.


In the right-moving sector all CFT building blocks have  $N=2$ worldsheet supersymmetry. This implies the existence of two operators
with simple current fusion rules: the worldsheet supercurrent $T_F$ and the spectra flow operator $S_F$. In general, the
internal CFT in the fermionic sector is itself built out of $N=2$ building blocks, that have such currents as well. 

In order to preserve right-moving world-sheet supersymmetry, the total supercurrent $T^{\rm st}_F+T_F^{\rm int}$ must have a well-defined periodicity, since it couples to the gravitino. Here, $T^{\rm st}_F=\psi^\mu \partial X_\mu$ is the world-sheet supercurrent in space-time and $T_F^{\rm int}$ is the supercurrent of the internal sector. Hence the allowed states will have the same spin structure in all the subsectors of the tensor product, namely the R (NS) sector of $SO(10)_1$ must be coupled to the R (NS) sector of the internal CFT. This result is achieved by an integer-spin simple current extension of the full right-moving algebra by the current given by the product of the supercurrents $T^{\rm st}_F\cdot T_F^{\rm int}$: it corresponds to projecting out all the combinations of mixed spin structures. When the internal CFT is a product of many sub-theories, as in the case of Gepner models, each with its own world-sheet supercurrent $T_{F,i}$, then one has to extend the full
  right-moving algebra by all the currents $T^{\rm st}_F\cdot  T_{F,i}$. In simple current language this means that we extend the chiral
algebra by all currents
\begin{equation}
\label{WS}
W_i= (0,\ldots,0,T_{F,i},0,\ldots,0;V)\ , \end{equation}
where we use a semi-colon to separate the internal and space-time part, and we use the standard notation $0,V,S,C$ for $SO(10)_1$ simple currents (or conjugacy classes).

A sufficient and necessary condition for space-time supersymmetry is the presence of a right-moving spin-$1$ chiral current transforming as an $SO(10)_1$ spinor. Hence this current must be equal to the product of the spinor $S$ of the $SO(10)_1$, which has spin $h=\frac{5}{8}$, times an operator $S^{\rm int}$ from the Ramond sector of the internal CFT $\mathscr{C}_9$, which must then have spin $h=\frac{3}{8}$. This last value saturates the chiral bound $h\geq\frac{c}{24}$ for the internal right-moving CFT which has central charge $c=9$, hence $S^{\rm int}$ corresponds to a Ramond ground state. 

Among the Ramond ground states, one is very special. $N=2$ supersymmetry possesses a one-parameter continuous automorphism of the algebra, known as spectral flow, which, when restricted to half-integer values of the parameter, changes the spin structures and maps Ramond fields to NS fields, hence uniquely relating fermionic to bosonic fields. In particular, under spectral flow, the NS field corresponding to the identity is mapped to a Ramond ground state which has $h=\frac{c}{24}$ and is called the spectral-flow operator. Not surprisingly, the spectral flow operator is related to the $N=1$ space-time supersymmetry charge. We will denote it as $S_F$.

In our set-up of four dimensional heterotic string theories, $N=1$ space-time supersymmetry is achieved again by a simple current extension. The current in question is the product of the space-time spin field $S_F^{\rm st}$ with $S_F^{\rm int}$, where $S_F^{\rm int}$ is the spectral-flow operator. If the internal CFT is built out of many factors, then $S_F^{\rm int}=\bigotimes_i S_{F,i}$, where $S_{F,i}$ is the spectral-flow operator in each factor. 
In simple current language, the space-time supersymmetry condition amounts to extending the chiral algebra of the CFT by
the simple current
\begin{equation}
\label{SS}
S_{\rm susy}=(S_{F,1},\ldots,S_{F,r};S)\ ,\end{equation}
where $r$ denotes the number of factors.
Obviously these simple current extensions must be closed under fusion, in combination with all world-sheet supersymmetry extensions
discussed above. Modular invariance of the final theory is then guaranteed by the simple current construction.

So far everything holds for any combination of superconformal $N=2$ building blocks. The only ones available in practice
(prior to this paper) are suitable combinations of free bosons and/or fermions, and $N=2$ minimal models. These are 
unitary finite-dimensional representations of the $N=2$ superconformal algebra, which exist only for  $c\leq3$.
 They are labelled by an integer $k$, in terms of which the central charge is 
\begin{equation}
c=\frac{3k}{k+2}\,.
\end{equation}
It is not a coincidence that $c$ is equal to the central charge of the $SU(2)_k$ affine Lie algebra. In fact, the $N=2$ minimal models admit a description in terms of the coset\footnote{According to this notation, $U(1)_N$ has $N$ primary fields, with $N$ always even.}
\begin{equation}
\frac{SU(2)_k \times U(1)_4}{U(1)_{2(k+2)}}\,.
\end{equation}
According to this decompositions, $N=2$ primary fields are products of parafermionic representations (which are primaries of the $\frac{SU(2)}{U(1)}$ algebra) times a bosonic exponential and the $S$ matrix is a product of the matrices $S^{SU(2)_k}$, $S^{U(1)_4}$ and $\left(S^{U(1)_{2(k+2)}}\right)^{-1}$. Representations are labelled by three integers $(l,m,s)$, where $l$ is an $SU(2)_k$ quantum number and $m$ and $s$ are $U(1)$ labels, satisfying the constraint $l+m+s=$even. The range is: $l=0,\dots,k$, $m=-k-1,\dots,k+2$, $s=-1,\dots,2$ ($s=0,2$ for NS sector, $s=\pm1$ for R sector). Moreover, fields are pairwise identified according to:
\begin{equation}
\phi_{l,m,s}\sim\phi_{k-l,m+k+2,s+2}\,,
\end{equation}
which is realized as a formal simple current extension.

Now consider the right-moving algebra of the heterotic string. The internal CFT $\mathscr{C}_9$ can be realized as a product of a sufficient number of $N=2$ minimal models such that 
\begin{equation}
\label{MinSum}
\sum_i^r\frac{3k_i}{k_i+2}=9\,,
\end{equation}
so the full algebra is
\begin{equation}
\bigotimes_i (N=2)_i \otimes (E_8)_1\otimes SO(10)_1
\end{equation}
and representations are labelled by
\begin{equation}
\bigotimes_i (l_i,m_i,s_i) \otimes (1) \otimes (s_0)\,.
\end{equation}
Observe that the $(E_8)_1$ algebra has only one representation, i.e. the identity, and it is often omitted in the product. Here $s_0$ denotes one of the four $SO(10)_1$ representations, $s_0=O,V,S,C$. As discussed above, we impose world-sheet and space-time supersymmetry by simple-current extensions. The world-sheet supercurrent for each $N=2$ minimal model is labelled by $T_{F,i}=(0,0,2)$ and the spectral-flow operator is $S_{F,i}=(0,1,1)$.  These are used in the world-sheet and space-time chiral algebra extensions (\ref{WS}) and 
(\ref{SS}).

These chiral algebra extensions are mandatory only in the fermionic sector. However, modular invariance does not allow
an extension in just one chiral sector. The most common way of dealing with this is to use exactly the same CFT in the 
left-moving sector, with exactly the same extensions. Of course any $N=2$ CFT is a special example of an $N=0$ CFT.
This construction leads to $(2,2)$ theories, with spectra analogous to Calabi-Yau compactifications, characterized by Hodge
number pairs and with a certain number of families in the $(27)$ of $E_6$. 
On the other hand, modular invariance is blind to most features of the CFT spectrum. It only
sees the modular group representations. This makes it possible to use in the left, bosonic, sector a different set of extension
currents than on the right. In particular one can replace the image of the space-time current by something else, thus breaking
$E_6$ to $SO(10)$. Furthermore one can break world-sheet supersymmetry in the bosonic sector. One can even go a step
further and break $SO(10)$ and $E_8$ to any subgroup, as long as this breaking can be restored by means of simple
currents. Those currents are then mandatory in the fermionic sector (since otherwise the bosonic string map cannot be used),
but can be replaced by isomorphic alternatives in the left sector. In general, we will call this class $(0,2)$ models. 

All the aforementioned possibilities will be considered in this paper, except $E_8$ breaking. The $SO(10)$ breaking
we consider is to $SU(3)\times SU(2) \times U(1)_{30} \times U(1)_{20}$, where the first three factors are the
standard model gauge groups with the standard $SU(5)$-GUT normalization for the $U(1)$. The fourth factor corresponds in
certain cases to $B-L$.


%

\section{Permutation orbifold of $N=2$ minimal models}
\label{Section: Permutation orbifold of N=2 minimal models}
In \cite{Maio:2010eu} the permutation orbifold of $N=2$ minimal models was studied. Extensions and permutations were performed in all possible orders and a nice structure was seen to arise, together with exceptional off-diagonal simple currents appearing in the extended orbifolds. In this section we recall the procedure of how to build a supersymmetric permutation orbifolds starting from $N=2$ minimal models. We will restrict ourselves to $\mathbb{Z}_2$ permutations, because a formalism to build permutation orbifold CFTs for higher
cyclic orders is not yet  available.

Consider the internal CFT $\mathscr{C}_9$ to be a tensor product of $r$ minimal models such that the total central charge is equal to $9$. We denote such a theory as\footnote{Note that here we mean the unextended tensor product. In particular, world-sheet supersymmetry 
extensions are not implied.}
\begin{equation}
(k_1,k_2,k_3\dots,k_r)\,,
\end{equation}
each $k_i$ parametrizing the $i^{\rm th}$ minimal model. Suppose that two of the $k_i$'s are equal: then the two corresponding minimal models are also identical and one can apply the orbifold mechanism to interchange them. We will label with brackets the block corresponding to the orbifold CFT: e.g. if $k_2=k_3$, then the permutation orbifold will be denoted by 
\begin{equation}
(k_1,\langle k_2,k_3\rangle\dots,k_r)\,.
\end{equation}
Multiple permutations are of course also possible. For convenience, we will follow the standard notation, used extensively in literature, of ordering the minimal models according to increasing level, namely $k_i\leq k_{i+1}$. Consequently, identical factors will always appear next to each other.
The orbifolded theory has the same central charge of the original one, namely $\sum_i^r c_i=9$, and hence can be used to build four dimensional string theories. 

Note that by $\langle k,k\rangle$ we mean the {\it supersymmetric} permutation orbifold, which, as explained in \cite{Maio:2010eu}, is obtained from the minimal model with level $k$ by first constructing the non-supersymmetric BHS orbifold (which we will denote as
$[ k,k]$), extending this CFT by
the anti-symmetric combination of the world-sheet supercurrent $(T_F,1)$, and resolving the fixed points occurring as a result
of that extension. This fixed point resolution promotes some fields to simple currents. All these simple currents will be used
to build MIPFs, using the general formalism presented in \cite{Kreuzer:1993tf}. 

Fixed point resolution enters the discussion at various points, and to prevent confusion we summarize here some relevant
facts. In the following we consider chains of extensions of the chiral algebra of a {\rm CFT}, and denote them as $({\rm CFT})_n$. Here
$({\rm CFT})_0$ is the original {\rm CFT}, $({\rm CFT})_1$ a first extension, $({\rm CFT})_2$ a second extension etc. In this process the chiral
algebra is enlarged in each step. The number of primary fields can decrease because some are projected out and others are
combined into new representations, but it can also increase due to fixed point resolution (apart from some special cases
the decrease usually wins over the increase). We are not assuming that each extension is itself ``indecomposable" ({\it i.e.} not
the result of several smaller extensions), but in practice the case of most interest will be a chain of extensions of order 2.
The following facts are important.
\begin{itemize}
\item{Simple currents $J$ are characterized by the identity $S_{0J}=S_{00}$, where $S$ is the modular transformation matrix.
For all other fields $i$, $S_{0i} > S_{00}$.}
\item{In an extension by a simple current of order $N$, the matrix elements $S_{0f}$ of fixed point fields are reduced by a factor of $N$. 
For this reason a fixed point field of $({\rm CFT})_n$ can be a simple current of $({\rm CFT})_{n+1}$. We will call these ``exceptional simple currents".}
\item{Exceptional simple currents can be used to build new MIPFs in $({\rm CFT})_{n+1}$, but such MIPFs are not simple
current MIPFs of $({\rm CFT})_n$. They are exceptional MIPFs.}
\item{If the fixed point resolution matrices of $({\rm CFT})_n$ are known, 
we can promote the exceptional simple currents of $({\rm CFT})_{n+1}$ to ordinary ones . This makes it possible to treat them on equal footing with
all other simple currents of $({\rm CFT})_{n+1}$.}
\item{Obviously, this process can be iterated: exceptional simple currents of $({\rm CFT})_{n+1}$ can themselves have fixed points, which can become
simple currents of $({\rm CFT})_{n+2}$.}
\item{If we know the fixed point resolution matrices of $({\rm CFT})_n$, we also know all the fixed point resolution matrices
of the ordinary simple currents of $({\rm CFT})_{n+1}$, but if the exceptional simple currents have fixed points, there is currently
no formalism available to determine their fixed point resolution matrices.}
\end{itemize}

In \cite{Maio:2009kb,Maio:2009cy,Maio:2009tg} we have developed a formalism for all fixed point resolution matrices
of the BHS permutation orbifolds. This plays the r\^ole of $({\rm CFT})_0$ in the foregoing. The supersymmetric permutation orbifold
$\langle k, k\rangle$ is $({\rm CFT})_1$. It always has exceptional simple currents, but only for $k=2~\hbox{mod}~4$ they have fixed 
points. As explained above, we cannot resolve these fixed points, but in heterotic spectrum computations
this is not necessary. This would be necessary if we want to go beyond spectrum computations to determine couplings. In spectrum
computations, fixed point fields $f$ appear in the partition function as character combinations of the form
\begin{equation}
N \bar\chi_f(\bar\tau)\chi_f(\tau), \  \  \  \  \ N>1 ,
\end{equation}
which is resolved into a certain number of distinct representations $(f,x)$ that contribute to the partition function as
\begin{equation}
\sum_x \bar\chi_{f,x}(\bar\tau)\chi_{f,x}(\tau), \ \ \ \ \chi_{f,x}(\tau) = m_x \chi_{f}(\tau)\ , \ \ \ \ \  \sum_x (m_x)^2=N \ .
\end{equation}
Note that for $N\geq4$ the last condition has several solutions, and to find out which one is the right one the
twist on the stabilizer of the fixed point must be determined \cite{Fuchs:1996dd}. However, here
we merely want to add up the values of $N$ for a left-right combination of interest, and the individual values of $m_x$ do not matter.

A few fields of the supersymmetric orbifold will be relevant in the following, all of untwisted type. They are:
\begin{itemize}
\item The symmetric representation of the spectral flow operator $(S_F,0)$, with $S_F=(0,1,1)$. It will be relevant to make the whole theory supersymmetric;
\item The world-sheet supercurrent of the supersymmetric orbifold, that we denote by $\langle 0,T_F \rangle$\footnote{Actually, since $\langle 0,T_F \rangle$ is a fixed point of $(T_F,1)$ in the unextended orbifold, there exist two fields $\langle 0,T_F \rangle_\alpha$ (with $\alpha=0,1$) in the supersymmetric orbifold corresponding to the two resolved fixed points. One can use any of them, since they produce the same CFT.};
\item The anti-symmetric representation of the identity, denoted by $(0,1)$. We will call it  the ``un-orbifold current" since the extension by this current undoes the orbifold, giving back the original tensor product. 
\end{itemize}

The un-orbifold current exists in the BHS orbifold $[k,k]$ as well as in the 
supersymmetric orbifold $\langle k, k\rangle$. Denoting extension currents by means of a subscript, we have
the following CFT relations
\begin{eqnarray*}
(k,k) &= [ k, k]_{\rm unorb}\\
(k,k)_{(T_F,T_F)} &= \langle k, k\rangle_{\rm unorb}
\end{eqnarray*}

In general, the full set of simple current MIPFs obtained from the permutation orbifold CFT  $(k_1,\langle k_2,k_3\rangle\dots,k_r)$ will
have a partial overlap with those of straight tensor product $(k_1, k_2,k_3,\ldots,k_r)$. Since the set of simple currents
of  $(k_1,\langle k_2,k_3\rangle\dots,k_r)$ includes the un-orbifold current one might expect that the latter set is entirely included in the
former. However, this is not quire correct, since the supersymmetric permutation orbifold has fewer simple currents than the
tensor product from which it originates, as explained above. In the extension chain,  $\langle k, k\rangle_{\rm unorb}$
is $({\rm CFT})_2$. In both steps in the chain

\begin{eqnarray*}
\hfill ({\rm CFT})_0 &= & [k,k] \\ &\downarrow\\  ({\rm CFT})_1 &=& \langle k, k\rangle\\ & \downarrow\\  ({\rm CFT})_2& =& \langle k, k\rangle_{\rm unorb} = (k,k)_{(T_F,T_F)}
\end{eqnarray*}
exceptional simple currents appear. Those of the first step are promoted to ordinary simple currents using
fixed point resolution in the BHS orbifold. We then work directly with $\langle k, k\rangle$ as a building block, but
by doing so we cannot use the exceptional simple currents emerging in the second step. In this case the
exceptional simple currents could be used by working with  $(k,k)_{(T_F,T_F)}$ directly, but then we are back in the unpermuted theory.
So the point is not that these MIPFs are unreachable, just that they cannot be reached using the simple currents of $\langle k, k\rangle $. Obviously, if we were to use a different exceptional simple current in the second extension, such that $({\rm CFT})_2$ is
a new, not previously known CFT  with exceptional simple currents, some of its MIPFs cannot be reached using simple current methods  
neither from $({\rm CFT})_1$ nor from $({\rm CFT})_2$. In all cases, one can try to derive such MIPFs explicitly as exceptional invariants, and
they can then be taken into account in heterotic spectrum computations, but this requires tedious and strongly case-dependent 
calculations. But in this paper we only consider simple current invariants, without any claim regarding completeness of the
set of MIPFs we obtain.

The phenomenon of  exceptional simple currents is nothing new, and occurs for example in the D-invariant
of $A_{1,4}$ (which is isomorphic to $A_{2,3}$), or the extension of the tensor products of two Ising models extended
by the product of the fermions (turning it into a free boson). 

The simplest explicit example occurs for $k=1$. In this case the discussion can be made a bit more explicit, since
the permutation orbifold is itself a minimal model, namely the one with $k=4$: 
\begin{eqnarray*}
 \langle 1,1\rangle  &=& (4)\\
 (1,1)_{(T_F,T_F)} &=&(4)_{\rm unorb} = (4)_D
\end{eqnarray*}
The minimal $k=1$ model  has 12 simple currents, and hence the tensor product $(1,1)$ has 144. To make
the tensor product world-sheet supersymmetric we have to extend it by $(T_F,T_F)$, reducing the number of simple currents by a factor
of four to 36. The $k=4$ minimal model has 24
simple currents. If we extend the $k=4$ minimal model by the un-orbifold current (which can be identified as such in the $\langle 1,1\rangle$
interpretation), these 24 original simple currents are reduced
to 12. Since the resulting CFT is isomorphic to $(1,1)_{(T_F,T_F)}$ there must be 24 additional simple currents. Indeed there are, but they 
are exceptional. They are related to the aforementioned exceptional currents in the D-invariant of $A_{1,4}$. This is
also the only example of exceptional simple currents in $N=2$ minimal models, and clearly in this case no MIPFs are
missed, since we can explicitly consider $(1,1)$ as well as $(4)_D$.  There might exist additional examples of exceptional
simple currents in tensor products of $N=2$ minimal models.

If the chiral algebra contains the un-orbifold current of a permutation orbifold, we obviously get nothing new. Therefore we
demand that this current is not in the chiral algebra. In general, it would be possible  to forbid it in either the left or the
right chiral algebra. This is already sufficient to find new cases. We do this, for example, with the $SU(5)$ extension
currents of the standard model, which are required in the right (fermionic string) chiral algebra, but not in the left one. However,
it turns out that the un-orbifold current is local with respect to all other simple currents.

In appendix \ref{Section: Simple current invariants} we prove a small theorem about simple current invariants. Consider a simple current  modular invariant partition function
\begin{equation}
Z(\tau,\bar{\tau})=\sum_{k,\,l}\bar{\chi}_k(\bar{\tau})M_{kl}\chi_l(\tau)\,.
\end{equation}
In the theorem it is shown that:
if a current $J$ that is local with respect to all currents used to construct the modular invariant appears on the right hand side (holomorphic sector) of the algebra, then
it will also appear on the left hand side (anti-holomorphic sector):
\begin{equation}
M_{0J}\neq 0 \qquad\Leftrightarrow \qquad M_{J0}\neq 0\,.
\end{equation}
Furthermore we show that the un-orbifold current is local with respect to all other currents.
Therefore the existence of the un-orbifold current on one side implies its existence also on the other side.
Hence it is sufficient to forbid its occurrence in either the left or the right sector.

However, there are a few cases where it cannot be forbidden at all, because it is generated by combinations
of world-sheet and space-time supersymmetry in the right (fermionic) sector, where such chiral algebra extensions 
are required. In general, a tensor product is extended by the currents $S_{\rm susy}$ and $W_i$, as explained in the previous 
section. 

If $k$ is even, the un-orbifold current does not appear on the orbit of the Ramond spinor current $S_F$, and hence can never
be generated. For arbitrary $k$ we have in the supersymmetric permutation orbifold
\begin{equation}
(S_F,0)^{2(k+2)}=
\left\{
\begin{array}{lc}
(0,0) & {\rm if}\,\,k\,\,{\rm even}\\
(0,1) & {\rm if}\,\,k\,\,{\rm odd}
\end{array}
\right.
\,,
\end{equation}
so that for $k$ odd one can obtain the un-orbifold current as a power of $S_F$. Note that instead of $2(k+2)$ one could use
any odd multiple of $2(k+2)$. In the tensor product $S_F$ is combined with the spinor currents of all the other factors,
which will be raised to the same power. Now note that in minimal models of level $k$ the following is true
\begin{equation}
(S_F)^{2(k+2)}=
\left\{
\begin{array}{lc}
0 & {\rm if}\,\,k\,\,{\rm even}\\
T_F & {\rm if}\,\,k\,\,{\rm odd}
\end{array}
\right.
\,.
\end{equation}
Furthermore, the value $2(k+2)$ is the first non-trivial power for which either the identity or the world-sheet supercurrent is
reached. It follows that if the tensor product contains a factor with $k_i$ even, 
the complete susy current  $(S_{F,1},\ldots,S_{F,r},S)$ must be raised to a power that is a multiple of four in order
to reach either the identity or a world-sheet supercurrent. This is true for minimal model factors as well as supersymmetric permutation
orbifolds $\langle k_i, k_i\rangle$. 

Consider then a tensor product $(k_1,\ldots,k_{m-1},\langle k_m,k_m\rangle,k_{m+1},\ldots,k_r)$. 
Take the susy current 
$$(S_{F,1},\ldots ,S_{F,m-1},(S_F,0),S_{F,m+1},\ldots S_{F,r};S)$$ 
to the power $2 M$, where $M$ is the smallest common multiple of $k_i+2$, for all $i$ (including $i=m$). If all $k_i$ are odd, 
this yields 
\begin{equation}
\label{Power}
(T_{F,1},\ldots,T_{F,m-1},(0,1),T_{F,m-1},\ldots T_{F,r};V)
\end{equation}
Since this is a power of an integer spin current, the
susy current, it must have integer spin. Therefore the number of $T_{F,i}$ must be odd. Indeed, it is not hard to show that
eqn. (\ref{MinSum}) can only be satisfied with all $k_i$ odd if the total number of factors, $r$, is odd. It then follows that all entries
$T_{F,i}$  as well as the representation $V$ of $SO(10)_1$ can be nullified by world-sheet supersymmetry. Hence it follows that
the un-orbifold current of $\langle k_m,k_m\rangle$ is automatically in the chiral algebra. It also follows that if one of the $k_i$ 
is even the un-orbifold current is {\it not} in the chiral algebra generated by $S_{\rm susy}$ and $W_i$. The same reasoning can be
applied to tensor products containing more than one permutation orbifold. The conclusion is that the un-orbifold currents of each factor separately are not generated by $S_{\rm susy}$ and $W_i$, but if all $k_i$ (in minimal models as well as the permutation
orbifolds) are odd, the combination $(0,\ldots,(0,1),\ldots,(0,1),\ldots,0;0)$, with an un-orbifold component in each permutation
orbifold, will automatically appear. Obviously, if there is more than one permutation orbifold factor this does not undo the permutation. 

The set of tensor combinations with only odd factors is rather limited, namely
\begin{eqnarray*}
&(1,1,1,1,1,1,1,1,1)\\
 &(3,3,3,3,3)\\
 &(1,3,3,3,13)\\
 &(1,1,7,7,7)\\
 &(1,1,5,5,19)\\
 &(1,1,3,13,13)
 \end{eqnarray*}
 We will not consider permutations of $k=1$, because $\langle1,1\rangle=4$, and hence nothing new can be found by allowing $\langle1,1\rangle$. Furthermore there is no need to
 consider any single permutations in the foregoing tensor products. However, we do expect the combinations 
 $(3,\langle 3,3\rangle,\langle 3,3\rangle)$, $(4,\langle 7,7\rangle,7)$ and 
 $(4,\langle 5,5\rangle,19)$ to yield something new.

For technical reasons in this paper we consider only permutations of minimal models having level $k \leq 10$: computing time and memory use become just too large for large $k$. Nevertheless, the interval $k\in[2,10]$ still covers almost all the standard Gepner models where at least two factors can be permuted.

\subsection{Permutations of permutations}
An additional thing that one can try to do (and which we can in principle do with our formalism, since we know all the relevant data that are needed) is to consider permutations of permutations. Permutations of permutations are possible only for a few Gepner models, because one would need to have a number of factors in the tensor product which is larger than four and with at least four identical minimal models. Out of the 168 possibilities, there are only a few combinations that have these properties. They are:
\begin{eqnarray}
 &(6,6,6,6)&\nonumber\\
 &(1,4,4,4,4)&\nonumber\\
 &(3,3,3,3,3)&\nonumber\\
 &(1,2,2,2,2,4)&\nonumber\\
 &(2,2,2,2,2,2)&
\end{eqnarray}
As before we restrict the $k>1$. Observe that the maximal level is $k=6$, so these cases are actually all the possibilities that one can consider and one can relax here our previous restriction to $k\leq10$.

The approach one should take is the following. Consider a block of four identical minimal models. As before we can permute the factors pairwise and obtain a tensor product of two larger blocks, but again identical. Hence we can permute them again and end up with only one big block which replaces the four ones that we started with:
\begin{displaymath}
\xymatrix{
(k,k,k,k) \ar[d] \\
(\langle k,k \rangle, \langle k,k\rangle) \ar[d]\\
\big\langle \langle k,k \rangle, \langle k,k\rangle\big\rangle
}
\end{displaymath}

Although straightforward, we will not perform this calculation in this paper. There are only very few cases to analyze, namely the five listed above, but, on the one hand, it is a pretty lengthy computation and on the other hand we do not expect drastically different spectra in comparison with normal permutations.

\section{Lifts}
\label{Section: Lifts}
In \cite{GatoRivera:2009yt} the authors describe a new method for constructing heterotic Gepner-like four-dimensional string theories out of $N=2$ minimal models. The method consists of replacing one $N=2$ minimal model together with the $E_8$ factor by a non-supersymmetric CFT with identical modular properties. Generically this method produces a spectrum with fewer massless states. Surprisingly, it is possible to get chiral spectra and gauge groups such as $SO(10)$, $SU(5)$ and other subgroups including the Standard Model. However, the most interesting feature is probably the abundant appearance of three-family models, which are very rare in standard Gepner models \cite{GatoRivera:2010gv}. Let us review how it is done in more detail, at least in the simplest case.

Start from the coset representation of the minimal model:
\begin{equation}
\frac{SU(2)_k\times U(1)_4}{U(1)_{2(k+2)}}\,,
\end{equation}
subject to field identification by the simple current $(J,2,k+2)$. Here $J$ is the simple current of the $SU(2)_k$ factor and the $U_N$ fields are labelled by their charges as $0,\dots,N-1$. 
The product of the $N=2$ minimal model and the $E_8$ factor is then
\begin{equation}
\left(\frac{SU(2)_k\times U(1)_4}{U(1)_{2(k+2)}}\right)_{(J,2,k+2)}\times E_8\,,
\end{equation}
where the brackets denote this identification. The next step is to remove the identification and mod out $E_8$ by $U(1)_{2(k+2)}$: the new CFT is then
\begin{equation}
SU(2)_k\times U(1)_4\times \frac{E_8}{U(1)_{2(k+2)}}\,.
\end{equation}
Finally we restore the identification by a standard order-2 current extension of the resulting CFT. This procedure works provided we can embed the $U(1)_{2(k+2)}$ factor into $E_8$. Some examples of how to embed $U(1)_{2(k+2)}$ into $E_8$ are given in \cite{GatoRivera:2009yt}. Finally, one can check explicitly that the modular $S$ and $T$ matrices are the same as for the $N=2$ 
minimal model times $E_8$, as they must be by construction. 
The resulting CFT is $SU(2)_k\times U(1)_4\times X_7$, where $X_7$ is the reminder of $E_8$ after dividing out $U(1)_{2(k+2)}$. $X_7$ has central charge $7$ and modular matrices $S$ and $T$ given by the complex conjugates of those of $U(1)_{2(k+2)}$ (since the ones of $E_8$ are trivial). Generically, this procedure raises the weights of the primaries in the new CFT, hence the name ``weight lifting''.

As it appears from above, the lifting of Gepner models is achieved by only a slight modification of standard Gepner models. All one has to
do is to shift the weights of some fields in the left-moving CFT by a certain integer, and replace the ground state dimensions by another,
usually larger, value. 
In \cite{GatoRivera:2009yt} a list of possible lifts is given for $N=2$ minimal models at level $k$. Not for any level there exists a lift and sometimes for fixed $k$ there are more lifts. When applied to standard Gepner models, a lot of new ``lifted'' Gepner models are generated. Notationally, if a Gepner model is denoted by $(k_1,\dots,k_i,\dots, k_r)$, the corresponding lifted model will be denoted by $(k_1,\dots,\hat{k}_i,\dots, k_r)$, where the lift is done on the $i^{\rm th}$ $N=2$ factor. If for a given $k$ there
exists more than one lift, we use a tilde to denote it.

In \cite{GatoRivera:2010fi} a different class of lifts was considered, the so called B-L lifts. In this case one replaces the $U(1)_{20}$ (with 20 primaries), that is the remainder of $SU(3)\times SU(2)\times U(1)$ embedded in $SO(10)$. In the Standard Model the abelian factor is the $U(1)_{Y}$ hypercharge (denoted also as $U(1)_{30}$, with 30 primaries). The $U(1)_{20}$ that we replace here corresponds to $B-L$, hence the name ``B-L lifting''. 

It is not possible to simply replace the $U(1)_{20}$ by an isomorphic CFT with 20 primaries, central charge $c=1$ and same modular $S$ and $T$ matrices, since all the $c=1$ CFTs are classified. Again, what one can do is to add the $E_8$ factor and replace the $E_8\times U(1)_{20}$ block, which has central charge $c=9$. 
As it turns out, there are only two possible B-L lifts, that we denote by $A$ and $B$. In terms of 
compactifications from ten dimensions these possibilities can be distinguished as follows. If one compactifies the $E_8\times E_8$ heterotic string one gets $SO(10)\times E_8$ in
four dimensions. The standard model can be embedded in $SO(10)$ (trivial lift, {\it i.e.} standard, unlifted $B-L$) or $E_8$ (lift A). If one compactifies the $SO(32)/\mathbb{Z}_2$ heterotic
string, one gets $SO(26)$, in which the standard model can then be embedded via an $SO(10)$ subgroup; this yields lift B. As explained in  \cite{GatoRivera:2010fi} both lift A and
lift B yield, perhaps counter-intuitively, chiral spectra. In the unlifted case, the number of families is typically a multiple of 6, and sometimes 2; for lift A, the family number quantization unit was found to be usually 1, whereas for lift B it was usually 2.

In this paper we will apply all these kinds of lifts to permuted Gepner models. This means that we make, when possible, all sorts of known lifts (namely, standard weight lifting and B-L lifting) for the $N=2$ factors that do \textit{not} belong to the sub-block(s) of the permutation orbifold. Note that permutations and lifting act independently: a given minimal
model factor is either unchanged, or lifted, or interchanged with another, identical factor. It may well be possible to construct lifted CFTs for the permutation orbifolds themselves,
but no examples are known, and they are in any case not obtainable by the methods of  \cite{GatoRivera:2009yt}, because there only a single minimal model factor is lifted. There
is one exception to this: there is one known simultaneous lift of two minimal model factors with $k=1$. Conceivably one could apply a permutation to those two identical
factors in combination with this lift. We have however not investigated this possibility.

\section{Comparison with known results}
\label{Section: Comparison}

To compare our results with previous ones on permutation orbifolds \cite{Klemm:1990df,Fuchs:1991vu}, it is  important to understand the
differences in these approaches. These authors first construct the basic Gepner model with all world-sheet and space-time supersymmetry
projections already in place in the left- as well as the right-moving sector.
 
They start from either the diagonal (A-type) invariants of all the minimal models, or the D and E-type
(exceptional) invariants. They then apply a cyclic permutation to the minimal model factors that are identical. They allow for additional
phase symmetries occurring in combination with the permutations. This combined operation
is applied to the full partition function. 

By contrast, we first build an $N=2$ permutation orbifold, then tensor it with other building blocks (either minimal models or other
$N=2$ permutation orbifolds), then impose world-sheet and space-time supersymmetry, but only on the fermionic sector, and consider general simple current modular invariants. 

So the differences can be summarized as follows
\begin{itemize}
\item{In \cite{Klemm:1990df,Fuchs:1991vu} general cyclic $\mathbb{Z}_L$ permutations are considered, while our results are limited to $L=2$.}
\item{In \cite{Klemm:1990df,Fuchs:1991vu} extra phases are modded out in combination with the permutations.}
\item{In \cite{Klemm:1990df,Fuchs:1991vu} permutations of D and E-invariants are considered.}
\item{We only consider permutations of factors with $2 \leq k \leq 10$.}
\item{We consider general simple current invariants.}
\item{We consider not only $(2,2)$ but also $(0,2)$ invariants and breaking of $SO(10)$.}
\end{itemize}

In order to make a comparison we will ignore the last point and focus on $(2,2)$ models. 
Since simple current invariants
include D-invariants as special cases, and because they involve monodromy phases of currents with respect to fields, one might expect
that at least part of the limitations in the second and third point are overcome. Exceptional invariants can be taken into account
in our method by multiplying the simple current modular matrix with an explicit exceptional modular matrix. Indeed, in standard Gepner
models we {\it have} taken them into account, and analysed the class of $(1,16^*,16^*,16^*)$ three-family models \cite{Gepner:1987hi}.
In the present case one could easily use exceptional invariants in those factors that are not permuted. To use permutations of exceptional
invariants we would first have to construct the exceptional MIPF explicitly in the permutation orbifold CFT, which can be done in principle
with a tedious computation. The first point is, however, much harder to overcome, because it would involve extending the BHS construction
to higher cyclic orders.  

Now let us see how the comparison works out in practice. In \cite{Fuchs:1991vu} a table is presented with
all models where cyclic permutations, phase symmetries and cyclic permutations together with phase symmetries have been modded out. For each model the authors give the number of generations $\rm n_{27}$, anti-generations $\rm n_{\overline{27}}$ and singlets $\rm n_{1}$. 
The first two numbers are equal to Hodge numbers of Calabi-Yau manifolds, namely $h_{21}= \rm n_{27}$ and 
$h_{11}=\rm n_{\overline{27}}$.
These quantities are first obtained by using modular invariance of the partition function of the cyclically-orbifolded Gepner models and are then compared with the same quantities derived by using topological arguments applied to the smooth Calabi-Yau manifold after that the singularities have been resolved. 
The number of families is specified by $\rm n_{gen}=n_{27}-n_{\overline{27}}$. The total number of singlets is strongly dependent on the multiplicities of the (descendants) states of the $N=2$ minimal models, which can be read off directly from the character expansions. The singlet number $\rm n_{1}$ turns out to be crucial for differentiating different models with equal $\rm n_{27}$ and $\rm n_{\overline{27}}$. Our comparison is based on these three numbers.  In table (\ref{HodgeNumbers}) we list the values 
we obtained for these three numbers in the cases we considered. Note that these are the numbers obtained without any
simple current extensions or automorphisms. The cases marked with a $*$ are $K_3 \times T_2$ type compactifications
with an $E_7$ spectrum; the numbers that are indicated are the ones obtained after decomposing $E_7$ to $E_6$.

In comparing the A-type invariants without phase symmetries, we get agreement, but in a somewhat unexpected way.
In \cite{Fuchs:1991vu} one-permutation models are not considered, because the authors argue that they always produce the same spectra as unpermuted Gepner models. However, 
in this paper we do manage to build one-permutation models as explained in section \ref{Section: Permutation orbifold of N=2 minimal models}. The only Gepner combinations for which the one-permutation models yield nothing new are the purely-odd combinations. 
Furthermore, the one-permutation orbifolds do indeed yield new results. For example, for the combinations $(2,2,2,2,\langle 2,2\rangle)$
the three numbers are $(90,0,284)$ as opposed to $(90,0,285)$ for the unpermuted case; for $(6,6,\langle6,6\rangle)$ we find 
$(106,2,364)$ as opposed to $(149,1,503)$; for $(\langle 3,3\rangle,10,58)$ we get $(75,27,392)$ as opposed to $(85,25,425)$. These
three example illustrate three distinct situations. In the first example, the only difference with the unpermuted case is that the number
of singlets is reduced by one. In the second example, the Hodge pair $(106,2)$ does occur for a non-trivial simple current invariant
of the tensor product $(6,6,6,6)$, namely $(6_A,6_A,6_A,6_D)$, but with 365 singlets instead of 364. In the
last example the Hodge pair $(75,27)$ does not occur for any simple current MIPF of $(3,3,10,58)$ (the only other combination that
occurs is $(53,41,401)$ plus the mirrors, so that even the Euler number of the permutation orbifold is new)\rlap.\footnote{For a complete
list of Hodge number and singlets of Gepner models see \cite{hodge}.}

In order to make a non-trivial comparison between
our spectra and those of \cite{Fuchs:1991vu} we have to look at Gepner models with {\it two} permutations. It turns out that our spectra (specified by $\rm n_{27}$, $\rm n_{\overline{27}}$ and $\rm n_{1}$) do agree with those of \cite{Fuchs:1991vu}. However, to get the full match, we always have to extend the model by one current. This current is (see section \ref{Section: Permutation orbifold of N=2 minimal models}) the double un-orbifold current, which has the un-orbifold current in each of the two factors corresponding to the permutation orbifold and the identity current in the remaining factors. 
Also in this case we already get new spectra even if we do {\it not} extend by this current. Consider for example $(\langle 6,6\rangle,\langle 6,6\rangle)$. As mentioned above, the $(6,6,6,6)$ gives $(149,1,503)$; the completely unextended spectrum we
get for the $(\langle 6,6\rangle,\langle 6,6\rangle)$ yields $(77,1,269)$; if we extend the two permutation orbifold CFTs by
the current combination $((0,1),(0,1))$ (where $(0,1)$ is the un-orbifold current) we find $(83,3,301)$, which is precisely the result quoted in \cite{Fuchs:1991vu}  for
the permutation orbifold. It is noteworthy that \cite{Fuchs:1991vu} lists a triplet $(77,1,271)$ for the combination
$(6_D,6_D,6_A,6_A)$, which from our perspective is a simple current invariant of $(6,6,6,6)$. Again we see two spectra with a minor
difference only in the number of singlets, which we will comment on below. In  one case we could not make
a comparison, because in \cite{Fuchs:1991vu} no result is listed for $(2,2,\langle 2,2\rangle,\langle 2,2\rangle)$ without extra phases.
In all other cases our results agree with \cite{Fuchs:1991vu}.
The need for extending by a combination of un-orbifold currents suggests that such currents are automatically generated or implicitly
present
in the
formalism used in \cite{Fuchs:1991vu}, for reasons we do not fully understand, but which are presumably related to an interchange
in the order of two operations: permutation and simple current extension. This is also consistent with the fact that these authors
find no new results for single permutations: if an un-orbifold current is automatically present in that case, one inevitably returns to the
unpermuted case. Note that for $(3,\langle3,3\rangle,\langle3,3\rangle)$ we have seen before that the separate un-orbifold current of
each permutation orbifold is automatically present in the chiral algebra, and hence so is the combination of the two. Therefore in this
case we do not have to extend by $(0,(0,1),(0,1))$ to find agreement with \cite{Fuchs:1991vu} because the extension is already automatically present.

Let us now compare the cases with extra phase symmetries.
In almost all cases, using the simple-current formalism, we recover for a given suitably-extended model the same Hodge numbers and the same singlet number as in those spectra where both the phase symmetry and the permutation symmetry have been modded out. In a sense, these phase symmetries correspond to simple current extensions or automorphisms. The only two exceptions, out of the many successful instances, both coming from the $2^6$ Gepner model (nr. 21 of Table II in \cite{Fuchs:1991vu}) with two permutations and phase symmetries, are
\begin{itemize}
\item (21)(43)56, 111100 ($\rm n_{27}=21$, $\rm n_{\overline{27}}=21$, $\rm n_{1}=180$, $\chi=0$),
\item (21)3(54)6, 333111 ($\rm n_{27}=44$, $\rm n_{\overline{27}}=8$, $\rm n_{1}=199$, $\chi=-72$),
\end{itemize}
where the first entry is the permutation orbifold and the second one is the phase symmetry. We were not able to find these two cases using our procedure.

There are a few other cases that we do not have, but for reasons that are easy to understand. Consider model nr. 168 in the same table. It correspond to the $6^4$ Gepner model. The double permutation that we reproduce is the one labelled as
\begin{itemize}
\item $6_A 6_A 6_A 6_A$: (21)(43) ($\rm n_{27}=83$, $\rm n_{\overline{27}}=3$, $\rm n_{1}=301$, $\chi=-160$).
\end{itemize}
The other two, with D invariants
\begin{itemize}
\item $6_A 6_A 6_D 6_D$: (21)(43) ($\rm n_{27}=45$, $\rm n_{\overline{27}}=1$, $\rm n_{1}=181$, $\chi=-88$),
\item $6_D 6_D 6_D 6_D$: (21)(43) ($\rm n_{27}=35$, $\rm n_{\overline{27}}=3$, $\rm n_{1}=154$, $\chi=-64$),
\end{itemize}
are not present. However these are not comparable with our $6^4$ since they come out of a different construction. In fact, the D invariant is obtained as a simple current automorphism of the $k=6$ Gepner models by the $SU(2)_k$ current $(k,0,0)$ (with $k=6$). This current has spin $h=\frac{k}{4}=\frac32$. In \cite{Fuchs:1991vu} the authors 
consider the permutation of two such $k=6$ models, each with such a simple current automorphism. 
This is different from what happens in this paper. Here, we immediately replace the block by its permutation orbifold; moreover, when we extend it by the current $(T_F,1)$ to build the supersymmetric permutation orbifold, the off-diagonal field $\langle (0,0,0)(6,0,0)\rangle$ with spin $h=\frac{3}{2}$ (the obvious candidate for creating the automorphism invariant) is not a simple current. We expect that the permutation
orbifold of two $6_D$ models is present as an exceptional invariant of $\langle6,6\rangle$.

The spectra mentioned in the last two paragraphs, that were present in \cite{Fuchs:1991vu} but absent in our results, might also be understood
as follows.  As explained in \ref{Section: Permutation orbifold of N=2 minimal models}, one may consider simple current extension
chains of the form
$$ {\rm (CFT)}_0 \rightarrow {\rm (CFT)}_1 \rightarrow {\rm (CFT)}_2 \rightarrow \ldots $$
In this chain, the supersymmetric permutation orbifold is ${\rm (CFT)}_1$. We can use all its simple currents to build
MIPFs, and in particular we find all simple current extensions ${\rm (CFT)}_2$. However there are situations where ${\rm (CFT)}_2$
itself has new simple currents that are exceptional, and whose orbits cannot be fully resolved because we do not have
the complete fixed point resolution formalism for ${\rm (CFT)}_1$ available. Therefore MIPFs generated by such second order
exceptional simple currents cannot be obtained. At best, one could try to get them by explicit computation as exceptional MIPFs
of ${\rm (CFT)}_1$. The problem of unresolvable fixed points occurs precisely for supersymmetric permutation orbifolds when
$k=2\mod 4$, and therefore might be relevant precisely in these examples. 

\LTcapwidth=14truecm
\begin{center}
\vskip .7truecm
\begin{longtable}{|c||c|c|c|}\caption{{\bf{Hodge data for permutation orbifolds of Gepner models.}}}\\
\hline
 \multicolumn{1}{|c||}{Tensor product}
& \multicolumn{1}{c|}{$h_{21}$}
& \multicolumn{1}{l|}{$h_{11}$ }
& \multicolumn{1}{c|}{Singlets}\\
\hline
\endfirsthead
\multicolumn{4}{c}%
{{\bfseries \tablename\ \thetable{} {\rm-- continued from previous page}}} \\
\hline 
 \multicolumn{1}{|c||}{model}
& \multicolumn{1}{c|}{$h_{21}$}
& \multicolumn{1}{l|}{$h_{11}$}
& \multicolumn{1}{c|}{Singlets}\\
\hline
\endhead
\hline \multicolumn{4}{|r|}{{Continued on next page}} \\ \hline
\endfoot
\hline \hline
\endlastfoot\hline
\label{HodgeNumbers}
$(1,1,1,1,1,1,\langle 2,2\rangle)$ &    $23^*$&     $23^*$&   177 \\ 
$(1,1,1,1,1,\langle 4,4\rangle)$ &    84 &      0 &   249 \\ 
$(1,1,1,1,\langle 10,10\rangle)$ &    57 &      9 &   248 \\ 
$(1,1,1,1,\langle 2,2\rangle,4)$ &    35 &     11 &   229 \\ 
$(1,1,1,\langle 2,2\rangle,2,2)$ &    $23^*$&     $23^*$&   175 \\ 
$(1,1,1,2,\langle 6,6\rangle)$ &    $23^*$&     $23^*$&   173 \\ 
$(1,1,1,\langle 4,4\rangle,4)$ &    73 &      1 &   242 \\ 
$(1,1,1,\langle 3,3\rangle,8)$ &    $23^*$&     $23^*$&   173 \\ 
$(1,1,1,\langle 2,2\rangle,\langle 2,2\rangle)$ &    $23^*$&     $23^*$&   173 \\ 
$(1,1,2,2,\langle 4,4\rangle)$ &    35 &     11 &   211 \\ 
$(1,1,\langle 2,2\rangle,2,10)$ &    46 &     10 &   234 \\ 
$(1,1,4,\langle 10,10\rangle)$ &    75 &      3 &   279 \\ 
$(1,1,\langle 6,6\rangle,10)$ &    37 &     13 &   211 \\ 
$(1,1,\langle 2,2\rangle,4,4)$ &    51 &      3 &   250 \\ 
$(1,1,\langle 2,2\rangle,\langle 4,4\rangle)$ &    35 &     11 &   209 \\ 
$(1,2,2,\langle 10,10\rangle)$ &    61 &      1 &   251 \\ 
$(1,\langle 2,2\rangle,2,2,4)$ &    61 &      1 &   260 \\ 
$(1,2,4,\langle 6,6\rangle)$ &    51 &      3 &   235 \\ 
$(1,2,\langle 4,4\rangle,10)$ &    62 &      2 &   241 \\ 
$(1,2,\langle 3,3\rangle,58)$ &    41 &     17 &   273 \\ 
$(1,\langle 4,4\rangle,4,4)$ &    84 &      0 &   279 \\ 
$(1,\langle 2,2\rangle,10,10)$ &    89 &      5 &   343 \\ 
$(1,\langle 3,3\rangle,4,8)$ &    41 &      5 &   219 \\ 
$(1,\langle 2,2\rangle,5,40)$ &    35 &     35 &   329 \\ 
$(1,\langle 2,2\rangle,6,22)$ &    68 &      8 &   297 \\ 
$(1,\langle 2,2\rangle,7,16)$ &    43 &     19 &   289 \\ 
$(1,\langle 2,2\rangle,8,13)$ &    27 &     27 &   249 \\ 
$(1,\langle 2,2\rangle,\langle 2,2\rangle,4)$ &    61 &      1 &   259 \\ 
$(\langle 2,2\rangle,2,2,2,2)$ &    90 &      0 &   284 \\ 
$(2,2,2,\langle 6,6\rangle)$ &    73 &      1 &   251 \\ 
$(2,2,\langle 4,4\rangle,4)$ &    51 &      3 &   242 \\ 
$(2,2,\langle 3,3\rangle,8)$ &    41 &      5 &   218 \\ 
$(2,2,\langle 2,2\rangle,\langle 2,2\rangle)$ &    90 &      0 &   283 \\ 
$(2,\langle 10,10\rangle,10)$ &   105 &      3 &   380 \\ 
$(2,\langle 8,8\rangle,18)$ &    79 &      7 &   322 \\ 
$(\langle 2,2\rangle,2,3,18)$ &    65 &      5 &   279 \\ 
$(2,\langle 7,7\rangle,34)$ &    63 &     15 &   312 \\ 
$(\langle 2,2\rangle,2,4,10)$ &    69 &      3 &   265 \\ 
$(\langle 2,2\rangle,2,6,6)$ &    86 &      2 &   297 \\ 
$(2,\langle 2,2\rangle,\langle 6,6\rangle)$ &    73 &      1 &   250 \\ 
$(3,\langle 6,6\rangle,18)$ &    51 &     11 &   254 \\ 
$(3,\langle 5,5\rangle,68)$ &    53 &     29 &   328 \\ 
$(3,\langle 8,8\rangle,8)$ &    99 &      3 &   346 \\ 
$(3,\langle 3,3\rangle,\langle 3,3\rangle)$ &    59 &      3 &   228 \\ 
$(4,4,\langle 10,10\rangle)$ &    94 &      4 &   334 \\ 
$(4,\langle 6,6\rangle,10)$ &    55 &      7 &   238 \\ 
$(4,\langle 5,5\rangle,19)$ &    41 &     17 &   238 \\ 
$(4,\langle 7,7\rangle,7)$ &    66 &      6 &   270 \\ 
$(\langle 5,5\rangle,5,12)$ &    83 &      5 &   308 \\ 
$(\langle 6,6\rangle,6,6)$ &   106 &      2 &   364 \\ 
$(\langle 4,4\rangle,10,10)$ &   101 &      5 &   370 \\ 
$(\langle 3,3\rangle,10,58)$ &    75 &     27 &   392 \\ 
$(\langle 3,3\rangle,12,33)$ &    47 &     31 &   306 \\ 
$(\langle 3,3\rangle,13,28)$ &    97 &     13 &   404 \\ 
$(\langle 3,3\rangle,18,18)$ &   125 &      9 &   490 \\ 
$(\langle 2,2\rangle,3,3,8)$ &    39 &     15 &   249 \\ 
$(\langle 2,2\rangle,4,4,4)$ &    60 &      6 &   285 \\ 
$(\langle 4,4\rangle,5,40)$ &    65 &     17 &   334 \\ 
$(\langle 4,4\rangle,6,22)$ &    70 &     10 &   304 \\ 
$(\langle 4,4\rangle,7,16)$ &    79 &      7 &   308 \\ 
$(\langle 4,4\rangle,8,13)$ &    48 &     12 &   242 \\ 
$(\langle 3,3\rangle,9,108)$ &    69 &     49 &   466 \\ 
$(\langle 6,6\rangle,\langle 6,6\rangle)$ &    77 &      1 &   269 \\ 
$(\langle 2,2\rangle,\langle 4,4\rangle,4)$ &    51 &      3 &   240 \\ 
$(\langle 2,2\rangle,\langle 3,3\rangle,8)$ &    41 &      5 &   216 \\ 
$(\langle 2,2\rangle,\langle 2,2\rangle,\langle 2,2\rangle)$ &    90 &      0 &   282 \\ 
$(1,\langle 2,2\rangle,\langle 10,10\rangle)$ &    61 &      1 &   250 \\ 
$(\langle 4,4\rangle,\langle 10,10\rangle)$ &    75 &      3 &   273 \\ 
$(1,\langle 4,4\rangle,\langle 4,4\rangle)$ &    73 &      1 &   234 \\ 
\end{longtable}
\end{center}

As already mentioned, the list of Hodge numbers and singlets in table (\ref{HodgeNumbers}) is obtained
without any simple current extensions other than those required to get a $(2,2)$ model. The complete list obtained
with arbitrary simple currents will be posted on the website of one of us \cite{hodge}. 

Although the 
results in table (\ref{HodgeNumbers}) are for $(2,2)$ models, the focus of the present paper was on $(0,2)$ models.
We can compare the results with those of \cite{GatoRivera:2010gv} and ask what permutation orbifolds add. Consider
first the set of $(0,2)$ models closest to $(2,2)$ models, namely those with an $E_6$ gauge symmetry. They are
characterized by the same three numbers ${\rm n}_{27}, {\rm n}_{\overline{27}}$ and  ${\rm n}_1$, but since there
is not necessarily a world-sheet supersymmetry in the bosonic sector they may not have a Calabi-Yau interpretation. For simplicity
we will refer to these  as ``pseudo Hodge pairs" and ``pseudo Hodge triplets".
In
the complete set of standard Gepner models without exceptional invariants we obtained a total of 1418\footnote{For the standard, unpermuted Gepner models, the number of genuine 
Hodge number pairs with world-sheet  supersymmetry in both sectors is 906. A list can be found on the website \cite{hodge}.
This includes Hodge numbers from asymmetric simple
current MIPFs not suitable for orientifold models.}
different pseudo Hodge  pairs
and 9604 different pseudo Hodge triplets. For the genuine permutation orbifolds (without extensions by un-orbifold currents)
these numbers are respectively  498 and 3830. Note that some permutation orbifolds with $k > 10$ were not considered.
How many of the permutation orbifold numbers are new? If we combine the data for pseudo Hodge pairs and remove identical ones,
we obtain a total of 1447 pseudo Hodge pairs, so that the total has increased by a mere 29. But if we look at pseudo Hodge triplets, the
increase is much more substantial. This number increases from 9604 to 12145, an increase of 2541 or about $26\%$. We tentatively
conclude that permutation orbifolds mainly give new points in existing moduli spaces. The following observation is further evidence
in that direction.

One remarkable feature of the permutation orbifold spectra is the occurrence of identical Hodge numbers and a number of singlets
that is almost the same. For example, in the set of permutation orbifolds obtained from the $(2,2,2,2,2,2)$ tensor product
we find spectra with (genuine) Hodge numbers $(90,0)$, and either 282, 283, 284 or 285 singlets.  A closer look at the spectrum
reveals what is going on here. We also compute the number of massless vector bosons in these spectra, and it turns out that
this is respectively 2,3,4 and 5 (in addition to those of $E_6$) in these cases. This is consistent with the occurrence of a Higgs mechanism that has made
one or more of the vector bosons heavy by absorbing the corresponding number of singlets. So apparently we are finding
points in the same moduli space, but with a vev for certain moduli fields so that some of the $U(1)$'s are removed. 
This is expected to occur in Gepner models, but it is nice to see this happen entirely within RCFT. The same observation was
made in \cite{Fuchs:1991vu}.
The reduction of the number of $U(1)$'s
by itself has a straightforward reason: each $N=2$ model has an intrinsic $U(1)$, and replacing two minimal models by a
permutation reduces the number of $U(1)$'s by 1. Hence the $(2,2,2,2,2,2)$ model generically has five $U(1)$'s (six, minus one 
combination
that becomes an $E_6$ Cartan-subalgebra generator), and $(\langle 2,2\rangle,\langle 2,2\rangle,\langle 2,2\rangle)$ generically has
only two. However, the number of vector bosons can be larger than that because the simple current MIPFs add extra
generators to the chiral algebra. Indeed, among the MIPFs of $(\langle 2,2\rangle,\langle 2,2\rangle,\langle 2,2\rangle)$ we
do not only find $(90,0,282,2)$ (where the last entry is the number of $U(1)$'s), but also $(90,0,283,3)$ and $(90,0,284,4)$.

\section{Results}
\label{Section: Results}

The CFT approach, based on simple currents extensions, turns out to be extremely powerful. Although we have considered in this paper only order-two permutations, the number of new modular invariant partition functions or, equivalently, the number of new spectra for each model is huge, in the order of a few thousands. Simple currents allow us to generate a huge number of four dimensional spectra.

Here we discuss the more phenomenological aspects of our results.
Conceptually this is very similar to work on unpermuted Gepner models presented in
 \cite{GatoRivera:2010gv,GatoRivera:2010xn,GatoRivera:2010fi}, to which we refer for
 more detailed  descriptions. 
 In these papers several, mostly empirical, observations were made regarding the resulting
 spectra.
 The main question of interest here is if these observations continue to hold as we extend the
 scope of RCFTs considered.

 \subsection{Gauge groups}
 
 We allow as gauge groups $SO(10)$ and seven rank 5 subgroups, namely the  Pati-Salam group
 $SU(4)\times SU(2)\times
 SU(2)$, the Georgi-Glashow GUT group $SU(5)\times U(1)$, two  global realizations of left-right symmetric algebra
 $SU(3)\times SU(2)\times SU(2)\times U(1)$, and three global realizations of the standard model algebra
  $SU(3)\times SU(2)\times U(1)\times U(1)$. Counted as Lie-algebras there are just five of them, but the last two
 come in several varieties when we describe them as CFT chiral algebras. These are distinguished by
 the fractionally charged (here ``charge" refers to unconfined electric charge)
 representations that are allowed. For the left-right algebra this can be either $\frac13$ or $\frac16$, 
 (we call these ``LR, Q=1/3" and ``LR, Q=1/6" respectively) 
 and for the standard model this can be $\frac12$, $\frac13$ or $\frac16$
 (SM, Q=1/2, 1/3 or 1/6).  In the string chiral algebra
 these different global realizations are distinguished by the presence of certain integer spin currents. If these
 currents have conformal weight one, they manifest themselves in the massless spectrum as extra gauge bosons. 
 This happens in particular for the highly desirable global group corresponding to the standard model with only
 integer unconfined electric charge. In this class of heterotic strings this necessarily implies an extension 
 of the standard model to (at least) $SU(5)$. 
Furthermore, if the standard model gauge group is extended to
$SU(5)$, this group cannot be broken by a field-theoretic Higgs mechanism, because the required Higgs scalar,
a $(24)$, cannot be massless in the heterotic string spectrum. A heterotic string spectrum contains either these
massless vector bosons, or fractionally charged states that forbid the former because they are non-local with
respect to them \cite{Schellekens:1989qb} (see also \cite{Wen:1985qj,Athanasiu:1988uj}).

These eight gauge groups are obtained as extensions of the affine Lie algebra $SU(3)_1\times SU(2)_1\times U_{30}$,
with a $U(1)$ normalization that gives rise to $SU(5)$-GUT type unification. In general, there is an additional $U(1)$ factor
that corresponds to a gauged $B-L$ symmetry in certain cases. In B-L lifted spectra this $U(1)$ is replaced by a non-abelian group.
In addition, the gauge group consists out of a $U(1)$ factor for each superconformal building block, which sometimes is extended
to a larger group, depending on the MIPF considered. There may also be extensions of the standard model gauge group outside
$SO(10)$, such as $E_6$ or trinification, $SU(3)^3$.
In standard Gepner models there is furthermore an unbroken $E_8$ factor, which
in lifted Gepner models is replaced by certain combinations of abelian and non-abelian groups. In scanning spectra we focus only
on the aforementioned  eight (extended) standard model groups. To identify their contribution in the figures we use the same
color-coding as in  \cite{GatoRivera:2010xn}, and the same names. These codes and names are shown in each figure.

\subsection{MIPF scanning}

Since it is essentially impossible to construct the complete set of distinct MIPFs, we use a random scan. This is done
by choosing 10.000 randomly chosen simple current subgroups ${\cal H}$ (see appendix A) generated by at most three simple currents.  
Furthermore, if the number of distinct torsion matrices $X$ is larger than 100, we make 100 random choices. The entire set
is guaranteed to be mirror symmetric, because for every given spectrum one  can always construct a mirror by multiplying the
MIPF with a simple current MIPF of $SO(10)_1$ that flips the chirality of all spinors. Note that this does not imply anything
about mirror symmetry of an underlying geometrical interpretation. It is a trivial operation on the spectrum that can however be used
to get some idea on the completeness of the scan. 
 
\subsection{Fractional Charges} 
 
Fractional charges can appear either in the form of chiral particles, or as vector-like particles  (where ``vector-like" is defined
with respect to the standard model gauge group) or only as massive particles, with masses of order the string scale. If a spectrum has
chiral fractionally charged particles, we reject it after counting it. In nearly all remaining cases the spectrum contains 
massless vector-like fractional charges (unless there is GUT unification). 
We regard such spectra as acceptable at this stage. Since no evidence for fractionally charged
particles exists in nature, with a limit of less than $10^{-20}$ in matter \cite{Perl:2009zz}, clearly these vector-like particles will have to acquire a mass.
Furthermore this will almost certainly have to be a huge (GUT scale or string scale) mass, since otherwise their abundance cannot
be credibly expected to be below the experimental limit. This can in principle happen if the vector-like particles couple to moduli that
get a vev. An analysis of existence of couplings is in principle doable in this class of models, although there may be some technical complications
in those cases where no fixed point resolution procedure is available at present (namely  the permutation orbifolds with
$k=2\ {\rm mod}\ 4$).  However,  this analysis is beyond the scope of this paper, and we treat spectra with 
vector-like fractional charges as valid candidates, for the time being. Just as in previous work \cite{GatoRivera:2010gv,GatoRivera:2010xn,GatoRivera:2010fi,FF2}, there are extremely rare
occurrences of spectra without any massless fractionally charged particles at all, but we only found examples with an
even number of families. Examples with three families were found in \cite{FF2} by scanning part of the
free-fermion landscape. In the context of orbifold models and Calabi-Yau compactifications, it is known that GUT breaking 
by modding out freely acting discrete symmetries leads to spectra without massless fractional charges (\cite{Witten:1985xc}; see
\cite{Blaszczyk:2009in} for a recent implementation of this idea in the context of the ``heterotic mini-landscape" \cite{Buchmuller:2005jr,Lebedev:2006kn,Lebedev:2008un}). While these models do fit the data on fractional charges, 
the question remains for which fundamental reasons such vacua are preferred over all others, especially if they are much rarer.

 In table (\ref{FreqTable}) we display how often four mutually exclusive types of spectra occur in the total sample,
 before distinguishing MIPFs. The types are: spectra with chiral, fractionally charged exotics, 
 chiral spectra with a GUT gauge group $SU(5)$ or $SO(10)$, non-chiral spectra (no exotics and no families),
 spectra with $N$ families and, massless $SU(5)$ vector bosons and vector-like fractionally charged exotics, and
 the same without massless fractionally charged exotics.  For comparison, we include some results
 based on the data of \cite{GatoRivera:2010gv,GatoRivera:2010xn,GatoRivera:2010fi}\rlap.\footnote{We thank the authors for making their
 raw data available to us.} All lines refer to Gepner models, except the one labelled ``free fermions".
 The results on free fermions are based on a special class that can be analysed with simple current in a 
 way analogous to Gepner models, as explained in \cite{GatoRivera:2010gv}. It does not represent the entire class of free fermionic models. For other work on this
 kind of construction, including three family models, see \cite{FF1,FF2} and references therein.
  
  \vskip 1.truecm
\begin{table}[h]
\begin{center}
~~~~~~~~~~~~~~~~\begin{tabular}{|l|c|c|c|c|c|cl} \hline
Type &  Chiral Exotics & GUT & Non-chiral & $N > 0$ fam. & No frac.\\ \hline \hline
Standard$^*$ & 37.4\% &32.7\% & 20.5 \% &  9.3\% & 0 \\
Standard, perm. & 29.7\% & 33.4 \% & 27.9 \% & 8.9\% & 0 \\ 
Free fermionic &  1.5\% &  2.9\% &   94.4\% &  1.1\% & 0.072\% \\ 
Lifted  & 28.3\% &  18.7\% &   51.9\% & 1.1\% & 0.00051\%\\ 
Lifted, perm. & 26.8\% & 8.9\% & 62.7 \% & 1.6\%  & 0.00078\% \\ 
$({\hbox{B-L}})_{\hbox{\footnotesize Type-A}}^*$ &  21.3\% &  28.0\% &  50.4 \% & 0.3\% & 0.00017\% \\ 
$({\hbox{B-L}})_{\hbox{\footnotesize Type-A}}$, perm. & 22.8\% & 8.1 \% & 69.1 \% & 0.03\% & 0 \\ 
$({\hbox{B-L}})_{\hbox{\footnotesize Type-B}}^*$ &  38.5\% &   8.7\% &  52.1\% &0.6\% & 0 \\ 
$({\hbox{B-L}})_{\hbox{\footnotesize Type-B}}$, perm. & 27.6\% & 7.3 \% & 65.0 \% & 0.1\% & 0 \\ 
\hline
 \end{tabular}
\caption{Relative frequency of various types of spectra. An asterisk indicates that exceptional minimal model MIPFs are included.}
\label{FreqTable}
\end{center}
\end{table}

In table (\ref{MIPFTable}) we specify the absolute number of distinct MIPFs (more precisely, distinct spectra,
based on the criteria spelled out in \cite{GatoRivera:2010gv,GatoRivera:2010xn,GatoRivera:2010fi}) with
non-chirally-exotic spectra. The column marked ``Total" specifies the total number of distinct spectra without
chiral exotics, the third column lists the number of distinct 3-family spectra and the last column the number of distinct
$N$-family spectra, in both cases regardless of gauge group and without modding out mirror symmetry.

 \vskip 1.truecm
\begin{table}[h]
\begin{center}
~~~~~~~~~~~~~~~~\begin{tabular}{|l|c|c|c|c|} \hline
Type &  Total & 3-family & $N$ family, $N>0$ \\ \hline \hline
Standard$^*$ & 927.100 & 1.220 & 369.089   \\
Standard, perm. & 245.821 & 0 & 64.085   \\ 
Free fermionic &  504.312 &   0 &  19.655   \\ 
Lifted  &     3.177.493 &  85.864 &  537.581 \\ 
Lifted, perm. & 601.452 & 4.702   &   54.926 \\ 
$({\hbox{B-L}})_{\hbox{\footnotesize Type-A}}^*$ &  445.978 &  24.203 & 155.425  \\ 
$({\hbox{B-L}})_{\hbox{\footnotesize Type-A}}$, perm. & 155.784 & 778 & 6.758   \\ 
$({\hbox{B-L}})_{\hbox{\footnotesize Type-B}}^*$ & 206949 &0 &   55917 \\ 
$({\hbox{B-L}})_{\hbox{\footnotesize Type-B}}$, perm. & 156.309 & 0 & 6.861  \\ 
\hline
 \end{tabular}
\caption{Total numbers of distinct spectra. }
\label{MIPFTable}
\end{center}
\end{table}


\subsection{Family number}
In this subsection we would like to say something about the distributions of the number of families emerging from the spectra of permuted Gepner models. The common features of all the different cases is that an even number of families is always more favourable than an odd one and these distributions decrease exponentially when the number of families increases.

Figure \ref{famplot_standard} shows the distribution of the number of families for permutation orbifolds of standard Gepner Models. 
All family numbers are even, as is the case for unpermuted Gepner models (we did not include exceptional MIPFs, which provides
the only way to get three families in standard Gepner models). The greatest common denominator $\Delta$ of the number of families
for a given tensor combination displays a similar behavior as  observed in \cite{GatoRivera:2010gv}. Two classes can be distinguished.
Either $\Delta=6$  (or in a few cases a multiple of $6$), or $\Delta=2$ (sometimes 4), but there are no MIPFs with a number of families 
that is a multiple of three. In other words, the set of family numbers occurring in these two classes have no overlap whatsoever. It follows
that in the second class there are no spectra with zero families.  An interesting example is\
$(3,\langle6,6\rangle,18)$. It has no spectra with chiral exotics, all spectra are chiral and have 4, 8, 16, 20, 28, 32, 40 or 56 families, of types
$SO(10)$, Pati-Salam, $SU(5)\times U(1)$ or $SM, Q=1/2$. If we compare this with the unpermuted Gepner model we find some
striking differences. In that case the same group types occur, but now there are spectra with chiral exotics, and the family quantization is
in units of 2, not 4. In  \cite{GatoRivera:2010gv} an intriguing observation was made regarding the occurrence of these two classes. 
The second class was found to occur if all values $k_i$ of the factors in the tensor product are divisible by 3. This observation
also holds for permutation orbifolds, if one uses the values of $k_i$ of the unpermuted theory.

In figure \ref{famplot_lift} we show the family distribution for lifted Gepner models. As expected, this distribution
looks a lot more favourable for three family models.
The number three appears with more or less the same order of magnitude as two or four. However, there are some clear peaks at even
family numbers, which were not visible in the analogous distribution for unpermuted Gepner models presented in  
\cite{GatoRivera:2010xn}. For this reason three families are still suppressed by a factor of 3 to 4 with respect to 2 or 4 families, which is
however a lot less than the two to three orders of magnitude observed in orientifold models \cite{Dijkstra:2004cc}. 

B-L lifts give similar results to those presented in  \cite{GatoRivera:2010fi}.
Figure \ref{famplot_liftA} contains the distribution of the number of families for permutation orbifolds of B-L lifted (lift A) Gepner models. Figure \ref{famplot_liftB} contains the same, but for the lift B. Here, odd numbers are completely absent. 
Note that certain group types (namely those without a ``$B-L$" type $U(1)$ factor) cannot occur in chiral spectra in these models, and
that in the  type that do occur the $U(1)$ is replaced by a non-abelian group.

\section{Conclusions}

In this paper we have considered $\mathbb{Z}_2$ permutation orbifolds of heterotic Gepner models. This should be viewed as an application of our previous paper \cite{Maio:2010eu} where $\mathbb{Z}_2$ permutations were studied for $N=2$ minimal models, which are the building blocks in Gepner construction.

Our main conclusion is that these new building blocks work as they should. They can be used on completely equal
footing with all other available ones, which are the $N=2$ minimal models and some free-fermionic building blocks.  We have
checked the combination with minimal models and found full agreement with previous results on permutation
orbifolds whenever they were available. The comparison did bring a few surprises, especially the fact that we were able
to get new spectra for single permutations,  where the old method of \cite{Fuchs:1991vu} gave nothing new.

We were able to go far beyond the old approach by finding many more $(2,2)$ models, as well as new $(0,2)$ models with
$SO(10)$ breaking. We combined all this with heterotic weight lifting and B-L lifting. The main conclusion is that in most respects
all observations 
concerning family number and fractional charges made for minimal models
continue to hold in this new class. Also in this case weight lifting greatly enhances  the
set of three family models in comparison to neighboring numbers. Although this appears to give some entirely new models
(Hodge number pairs that were not seen before), we found additional evidence for the observation of \cite{Fuchs:1991vu} that
many of these models look like additional rational points in existing moduli spaces. 


A next step in the exploration of these models is the study of couplings. An important problem to address is the possibility of lifting
the vector-like states, especially the omnipresent fractionally charged ones, by giving vevs to moduli. In some cases, namely
the permutation orbifolds with $k=2\mod4$, one may run into the problem of unresolved fixed points for certain couplings. This
problem did not arise here because we only studied spectra.
But there are plenty of cases where this restriction does not apply. The fact that all building blocks  with $k\not=\!2\mod4$ are
on equal footing with minimal models even with respect to fixed points
 should make it possible to develop an algorithm for determining couplings all at once for all cases. This is certainly a problem we
 intend to study in the near future.

Another application of our permutation orbifold CFTs
 could be to study the case of orientifold models. However, there we will face the same problem with fixed point resolution 
 at an earlier stage, since
 the fixed point resolution matrices enter in the formula for boundary  coefficients. We hope to address this issue in the future.

\section*{Acknowledgments}
This research is supported by the Dutch Foundation for Fundamental Research of Matter (FOM)
as part of the program STQG (String Theory and Quantum Gravity, FP 57). This work has been partially 
supported by funding of the Spanish Ministerio de Ciencia e Innovaci\'on, Research Project
FPA2008-02968, and by the Project CONSOLIDER-INGENIO 2010, Programme CPAN (CSD2007-00042).

\appendix

\section{Simple current invariants}
\label{Section: Simple current invariants}
Consider a simple current $J$ or order $N$, i.e. $J^N=1$. Define the \textit{monodromy parameter} $r$ as
\begin{equation}
h_J=\frac{r(N-1)}{2N}\quad {\rm mod}\,\,\mathbb{Z}\,.
\end{equation}
Also, define the \textit{monodromy charge} $Q_J(\Phi)$ of $\Phi$ w.r.t. $J$ as
\begin{equation}
Q_J(\Phi)=h_J+h_{\Phi}-h_{J\phi}\quad {\rm mod}\,\,\mathbb{Z}\,.
\end{equation}
The monodromy charge takes values $t/N$, with $t\in\mathbb{Z}$. The current $J$ organizes fields into orbits $(\Phi,J\Phi,\dots,J^d\Phi)$, where $d$ (not necessarily equal to $N$) is a divisor of $N$. On each orbit the monodromy charge is
\begin{equation}
Q_J(J^n\Phi)=\frac{t+rn}{N}\quad {\rm mod}\,\,\mathbb{Z}\,.
\end{equation}

If a simple current $J$ exists in a (rational) CFT, and if it satisfies the condition that $N$ times its conformal weight is an integer\rlap,\footnote{This is sometimes called the ``effective center condition", and eliminates for example the odd level simple currents of $A_1$, which have
order two, but quarter-integer spins.}
then it is known how to associate a modular invariant partition function to it. Suppose that the current $J$ has integer spin and order $N$. Then a MIPF is given by
\begin{equation}
\label{MIPF}
Z(\tau,\bar{\tau})=\sum_{k,\,l}\bar{\chi}_k(\bar{\tau})M_{kl}(J)\chi_l(\tau)\,.
\end{equation}
One way of expressing $M_{kl}(J)$ is \cite{Kreuzer:1993tf,Schellekens:1989wx}:
\begin{equation}
\label{M_kl}
M_{kl}(J)=\sum_{p=1}^N \delta(\Phi_k,J^p \Phi_l)\cdot\delta^1(\hat{Q}_J(\Phi_k)+\hat{Q}_J(\Phi_l))\,
\end{equation}
where $\delta^1(x)=1$ for $x=0$ mod $\mathbb{Z}$ and $\hat{Q}$ is defined on $J$-orbits as
\begin{equation}
\hat{Q}_J(J^n\Phi)=\frac{(t+rn)}{2N} \quad{\rm mod}\,\,\mathbb{Z}\,.
\end{equation}
Morally speaking, $\hat{Q}$ is half the monodromy charge. Formula (\ref{M_kl}) defines a modular invariant partition functions, since it commutes with the $S$ and $T$ modular matrices, as shown in \cite{Schellekens:1989dq}.
The set of all the simple currents forms an abelian group $\mathcal{G}$ under fusion multiplication. It is always a product of cyclic factors generated by a (conventionally chosen) complete subset of independent simple currents.

The foregoing associates a modular invariant partition function with a single simple current. One can construct even more of them
by multiplying the matrices $M$. The most general simple current MIPF associated with a given subgroup of $\mathcal{G}$ can be
obtained as  follows \cite{Kreuzer:1993tf,GatoRivera:2010gv}. Choose a subgroup of $\mathcal{G}$ denoted $\mathcal{H}$, such that
each element satisfies the effective center condition $N h_J \in \mathbb{Z}$.  Its generators are simple currents $J_s$, $s=1,\dots,k$ for some $k$. 
They have relative monodromies $Q_{J_s}(J_t)=R_{st}$. Take any matrix $X$ that satisfies the equation
\begin{equation}
X+X^T=R\,.
\end{equation}
The matrix $X$ (called the {\it torsion matrix}) determines the multiplicities $M_{ij}$ according to
\begin{equation}
\label{M=nr of sols}
M_{ij}(\mathcal{H},X)={\rm nr.\,\,of\,\,solutions}\,\,K\,\,{\rm to\,\,the\,\,conditions:}
\end{equation}
\begin{itemize}
\item $j=Ki$, $K\in\mathcal{H}$.
\item $Q_M(i)+X(M,K)=0$  mod 1 for all $M\in\mathcal{H}$\,.
\end{itemize}
Here $X(K,J)$ is defined in terms of the generating current $J_s$ as
\begin{equation}
 X(K,J)\equiv \sum_{s,t}n_s m_t X_{st}\,,
\end{equation}
with $K=(J_1)^{n_1}\dots(J_k)^{n_k}$ and $J=(J_1)^{m_1}\dots(J_k)^{m_k}$.


\subsection{A small theorem}
In this subsection we prove the following theorem.

\begin{theorem}
The following statements are true.\\
i) If a simple current $J$ is local w.r.t. any other current $K$, i.e. $Q_K(J)=0$ (mod $\mathbb{Z}$), then $M_{JJ}(K)\neq 0$.\\
ii) For a simple current $J$, which is local w.r.t. any other current $K$, $M_{J0}(K)=M_{0J}(K)$. 
In particular, if $M_{J0}(K)\neq 0$, then also $M_{0J}(K)\neq 0$.
\end{theorem}
\begin{proof}
For the proof we use the statement (\ref{M=nr of sols}).\\
Let us start with {\it i)} and consider $M_{JJ}(K,X)$. The first condition in (\ref{M=nr of sols}) has only one solution, namely $K=0$. The second condition reads
\begin{equation}
Q_M(J)+X(M,0)=0 \qquad \forall M
\end{equation}
and is always true, because the two terms vanish separately. This proves that $M_{JJ}(K)\neq 0$.

Point {\it ii)} goes as follows. Consider $M_{0J^c}(K,X)$. There is again only one solution to the first condition, namely $K=J^c$. The second condition reads
\begin{equation}
Q_M(0)+X(M,J^c)=0\,.
\end{equation}
The first term vanishes by hypothesis, while the second is either zero (in which case $M_{0J^c}(K,X)\neq0$) or non-zero (in which case $M_{0J^c}(K,X)=0$). \\
Similarly, look at $M_{J0}(K,X)$. There is again only one solution to the first condition, namely $K=J^c$. The second condition reads
\begin{equation}
Q_M(0)+X(M,J^c)=0\,.
\end{equation}
The first term vanishes by hypothesis, while the second is either zero or non-zero. In any case, the same condition holds for both $M_{0J^c}(K,X)$ and $M_{J0}(K,X)$. This implies that $M_{J0}(K,X)=M_{0J^c}(K,X)$ (note that these matrix elements can only be 0 or 1).
By closure of the algebra, and because $J^c$ is always a
power of $J$, we may replace $J^c$ by $J$ in this relation.
\end{proof}

Consider now the permutation orbifold. This theorem applies in particular to the un-orbifold current $J=(0,1)$ when coupled to any other current $K$, which is either a standard (diagonal) or an exceptional (off-diagonal) one. In fact, using the same procedure as we did in \cite{Maio:2009kb} to compute the simple current and fixed point structure of the permutation orbifold, one can show that
\begin{eqnarray}
N_{(J,\psi)\langle pq\rangle}^{\phantom{(J,\psi)\langle pq\rangle}\langle p'q'\rangle} &=&
N_{Jp}^{\phantom{Jp}p'}N_{Jq}^{\phantom{Jq}q'}+N_{Jp}^{\phantom{Jp}q'}N_{Jq}^{\phantom{Jq}p'}\,,\nonumber\\
N_{(J,\psi)(i,\chi)}^{\phantom{(J,\psi)(i,\chi)}(i',\chi')} &=&
\frac{1}{2} N_{Ji}^{\phantom{Ji}i'} (N_{Ji}^{\phantom{Ji}i'}+e^{i\pi(\psi+\chi-\chi')})\,.\nonumber
\end{eqnarray}
Hence, for the current $(J,\psi)=(0,1)$, we have
\begin{equation}
N_{(0,1)\langle pq\rangle}^{\phantom{(0,1)\langle pq\rangle}\langle p'q'\rangle} =
\delta_p^{p'}\delta_q^{q'}+\delta_p^{q'}\delta_q^{p'}=\delta_p^{p'}\delta_q^{q'}\,,\nonumber
\end{equation}
namely $\langle pq\rangle$ must be fixed by $(0,1)$ in order for this to be non-zero (recall that $p<q$ and $p'<q'$), and
\begin{equation}
N_{(0,1)(i,\chi)}^{\phantom{(0,1)(i,\chi)}(i',\chi')} =
\frac{1}{2} \delta_i^{i'}(\delta_i^{i'}-e^{i\pi(\chi-\chi')})\,\nonumber
\end{equation}
which is non-zero only if $i=i'$ and $\chi\neq\chi'$ (recall that we can think of $\chi$ as defined mod $2$). Equivalently, in the fusion language:
\begin{equation}
(0,1)\cdot \langle pq\rangle=\langle pq\rangle\quad,\qquad
(0,1)\cdot (i,\chi)=(i,\chi+1)\,.
\end{equation}
This implies that $(0,1)$ has zero monodromy charge w.r.t. any other current, since
\begin{eqnarray}
Q_{\langle pq\rangle}\big( (0,1) \big) &=&
h_{\langle pq\rangle}+h_{(0,1)}-h_{\langle pq\rangle}=0\quad{\rm mod}\,\,\mathbb{Z}\,,\nonumber\\
Q_{(i,\chi)}\big( (0,1) \big) &=&
h_{(i,\chi)}+h_{(0,1)}-h_{(i,\chi+1)}=0\quad{\rm mod}\,\,\mathbb{Z}\,.\nonumber
\end{eqnarray}

Now, the un-orbifold current $(0,1)$ has order two, hence $J^c=J$ and $M_{J0}(K,X)=M_{0J}(K,X)$. This also implies that its existence on left-moving sector is guaranteed by its existence on the right-moving sector (and vice-versa).

\subsection{Summary of results}

Here we present four tables summarizing the results on the number of families for the
standard, heterotic weight lifted, B-L lifted (lift A) and B-L lifted (lift B) cases. These tables contain
information about spectra in which the un-orbifold current is not allowed in the chiral algebra. This
means that these are genuine permutation orbifold spectra. By inspection, we do indeed find that these
spectra are usually different than those obtained in the unpermuted case.
In the columns we specify
respectively the tensor combination  the greatest common divisor $\Delta$ of the number of families for all MIPFs of
the tensor product and the maximal net number of families encountered. In the next column we indicate which of the seven
$SO(10)$ subgroups occur, with the labelling 
\begin{itemize}
\item{0: SM,  Q=1/6} \item{ 1: SM,  Q=1/3} \item{ 2: SM,  Q=1/2} \item{ 3: LR, Q=1/6} \item{4: $SU(5)\times U(1)$} \item{ 
5: LR, Q=1/3} \item{ 6: Pati-Salam. }
\end{itemize}
Since $SO(10)$ can always occur there is no need to indicate it. In \cite{GatoRivera:2010gv} a simple criterion was derived
to determine which subgroups can occur in each standard Gepner model. The allowed subgroups for permutation orbifolds
of Gepner models are a subset of these. In some cases, such as $(\langle 5,5\rangle,5,12)$ some of the subgroups cannot be realized.
In the column labelled ``Exotics" we indicate if, for a given tensor product, spectra with chiral exotics occur. Note that 
in most cases the absence of such spectra is a consequence of the fact that only GUT gauge groups occur. In the
next columns we list
the number of distinct  three family and in column 6
the number of distinct $N$-family $(N>0)$ spectra.
In these tables
only cases with $\Delta > 0$ are shown. If a permutation orbifold seems to be missing, than either it is a  permutation
for $k > 10$, or it is a purely odd tensor product for which all permutations are trivial, or it has only non-chiral spectra and
hence $\Delta=0$ and there are no chiral exotics.  
In column 1 of the second table, $\langle \cal{A},\cal{A}\rangle$  denotes the permutation
orbifold of CFT ${\cal A}$, a hat indicates the lifted CFT, and a tilde indicates the second lift of a CFT.
It turns out that in the only permutation orbifold with two distinct
lifts of the same factor, $(\hat 5,\langle5,5\rangle,12)$ and  $(\tilde 5,\langle5,5\rangle,12)$, $\Delta=0$ in both cases,
which is why a tilde never occurs in the tables. The last column indicates which percentage of the spectra has no mirror.
Since mirror symmetry is exact in the full set, this gives an indication of how close our random scan is to a full enumeration.

\LTcapwidth=14truecm
\begin{center}
\vskip .7truecm
\begin{longtable}{|c||c|c|c|c|c|c|c|}\caption{{\bf{Results for standard Gepner models}}}\\
\hline
 \multicolumn{1}{|c||}{model}
& \multicolumn{1}{c|}{$\Delta$}
& \multicolumn{1}{l|}{Max. }
& \multicolumn{1}{c|}{Groups }
& \multicolumn{1}{l|}{Exotics }
& \multicolumn{1}{c|}{3 family}
& \multicolumn{1}{c|}{$N$ fam.}  
& \multicolumn{1}{c|}{Missing} \\ 
\hline
\endfirsthead
\multicolumn{8}{c}%
{{\bfseries \tablename\ \thetable{} {\rm-- continued from previous page}}} \\
\hline 
 \multicolumn{1}{|c||}{model}
& \multicolumn{1}{c|}{$\Delta$}
& \multicolumn{1}{l|}{Max. }
& \multicolumn{1}{c|}{Groups }
& \multicolumn{1}{l|}{Exotics}
& \multicolumn{1}{c|}{3 family}
& \multicolumn{1}{c|}{$N$ fam.}  
& \multicolumn{1}{c|}{Missing} \\ 
\hline
\endhead
\hline \multicolumn{8}{|r|}{{Continued on next page}} \\ \hline
\endfoot
\hline \hline
\endlastfoot\hline
\label{StandardSummary}
$(1,1,1,1,1,\langle 4,4\rangle)$ & 6 & 84 & 3,5,6 & Yes & 0 & 342 &    6.14\% \\ 
$(1,1,1,1,\langle 10,10\rangle)$ & 6 & 48 & 3,5,6 & Yes & 0 & 124 &    4.84\% \\ 
$(1,1,1,1,\langle 2,2\rangle,4)$ & 6 & 48 & 3,5,6 & Yes & 0 & 75 &    6.67\% \\ 
$(1,1,1,\langle 4,4\rangle,4)$ & 6 & 84 & 3,5,6 & Yes & 0 & 2717 &   22.89\% \\ 
$(1,1,2,2,\langle 4,4\rangle)$ & 6 & 24 & 3,5,6 & Yes & 0 & 106 &    0.00\% \\ 
$(1,1,\langle 2,2\rangle,2,10)$ & 6 & 48 & 3,5,6 & Yes & 0 & 662 &    7.70\% \\ 
$(1,1,4,\langle 10,10\rangle)$ & 6 & 72 & 3,5,6 & Yes & 0 & 493 &    7.10\% \\ 
$(1,1,\langle 6,6\rangle,10)$ & 12 & 24 & 3,5,6 & Yes & 0 & 63 &    0.00\% \\ 
$(1,1,\langle 2,2\rangle,4,4)$ & 6 & 48 & 3,5,6 & Yes & 0 & 226 &    6.19\% \\ 
$(1,1,\langle 2,2\rangle,\langle 4,4\rangle)$ & 12 & 24 & 3,5,6 & Yes & 0 & 73 &    6.85\% \\ 
$(1,2,2,\langle 10,10\rangle)$ & 6 & 60 & 3,5,6 & Yes & 0 & 191 &    4.71\% \\ 
$(1,\langle 2,2\rangle,2,2,4)$ & 12 & 60 & 3,5,6 & Yes & 0 & 363 &    3.31\% \\ 
$(1,2,4,\langle 6,6\rangle)$ & 12 & 48 & 3,5,6 & Yes & 0 & 57 &    3.51\% \\ 
$(1,2,\langle 4,4\rangle,10)$ & 6 & 60 & 3,5,6 & Yes & 0 & 1605 &   14.08\% \\ 
$(1,2,\langle 3,3\rangle,58)$ & 6 & 24 & 0,1,2,3,4,5,6 & Yes & 0 & 102 &    0.00\% \\ 
$(1,\langle 4,4\rangle,4,4)$ & 6 & 84 & 3,5,6 & Yes & 0 & 5605 &    6.57\% \\ 
$(1,\langle 2,2\rangle,10,10)$ & 6 & 84 & 3,5,6 & Yes & 0 & 989 &    6.47\% \\ 
$(1,\langle 3,3\rangle,4,8)$ & 6 & 36 & 0,1,2,3,4,5,6 & Yes & 0 & 37 &    0.00\% \\ 
$(1,\langle 2,2\rangle,6,22)$ & 6 & 60 & 3,5,6 & Yes & 0 & 985 &    3.25\% \\ 
$(1,\langle 2,2\rangle,7,16)$ & 12 & 48 & 3,5,6 & Yes & 0 & 41 &    0.00\% \\ 
$(1,\langle 2,2\rangle,\langle 2,2\rangle,4)$ & 12 & 60 & 3,5,6 & Yes & 0 & 165 &    0.61\% \\ 
$(\langle 2,2\rangle,2,2,2,2)$ & 6 & 90 & 6 & Yes & 0 & 1849 &    5.19\% \\ 
$(2,2,2,\langle 6,6\rangle)$ & 12 & 72 & 6 & Yes & 0 & 245 &    0.00\% \\ 
$(2,2,\langle 4,4\rangle,4)$ & 6 & 48 & 3,5,6 & Yes & 0 & 250 &    0.00\% \\ 
$(2,2,\langle 3,3\rangle,8)$ & 6 & 36 & 2,4,6 & Yes & 0 & 55 &    0.00\% \\ 
$(2,2,\langle 2,2\rangle,\langle 2,2\rangle)$ & 6 & 90 & 6 & Yes & 0 & 1580 &    1.58\% \\ 
$(2,\langle 10,10\rangle,10)$ & 6 & 102 & 3,5,6 & Yes & 0 & 328 &    0.00\% \\ 
$(2,\langle 8,8\rangle,18)$ & 6 & 72 & 2,4,6 & Yes & 0 & 316 &    0.00\% \\ 
$(\langle 2,2\rangle,2,3,18)$ & 6 & 60 & 2,4,6 & Yes & 0 & 780 &    4.36\% \\ 
$(2,\langle 7,7\rangle,34)$ & 12 & 48 & 3,5,6 & Yes & 0 & 9 &    0.00\% \\ 
$(\langle 2,2\rangle,2,4,10)$ & 6 & 66 & 3,5,6 & Yes & 0 & 1550 &    3.81\% \\ 
$(\langle 2,2\rangle,2,6,6)$ & 6 & 84 & 6 & Yes & 0 & 1735 &    3.80\% \\ 
$(2,\langle 2,2\rangle,\langle 6,6\rangle)$ & 12 & 72 & $SO(10)$ only & No & 0 & 219 &    0.00\% \\ 
$(3,\langle 6,6\rangle,18)$ & 4 & 56 & 2,4,6 & No & 0 & 232 &    0.00\% \\ 
$(3,\langle 5,5\rangle,68)$ & 24 & 24 & 4 & No & 0 & 18 &    0.00\% \\ 
$(3,\langle 8,8\rangle,8)$ & 6 & 96 & 2,4,6 & Yes & 0 & 1909 &    1.41\% \\ 
$(3,\langle 3,3\rangle,\langle 3,3\rangle)$ & 2 & 56 & 4 & No & 0 & 126 &    0.00\% \\ 
$(4,4,\langle 10,10\rangle)$ & 6 & 90 & 5 & Yes & 0 & 188 &    0.00\% \\ 
$(4,\langle 6,6\rangle,10)$ & 12 & 48 & 5 & Yes & 0 & 70 &    0.00\% \\ 
$(4,\langle 5,5\rangle,19)$ & 12 & 24 & 5 & Yes & 0 & 6 &    0.00\% \\ 
$(4,\langle 7,7\rangle,7)$ & 12 & 60 & 5 & Yes & 0 & 11 &    0.00\% \\ 
$(\langle 5,5\rangle,5,12)$ & 6 & 78 & $SO(10)$ only & No & 0 & 44 &    0.00\% \\ 
$(\langle 6,6\rangle,6,6)$ & 2 & 104 & 6 & Yes & 0 & 1230 &    0.00\% \\ 
$(\langle 4,4\rangle,10,10)$ & 6 & 96 & 3,5,6 & Yes & 0 & 693 &    0.72\% \\ 
$(\langle 3,3\rangle,10,58)$ & 6 & 60 & 0,1,2,3,4,5,6 & Yes & 0 & 97 &    0.00\% \\ 
$(\langle 3,3\rangle,12,33)$ & 2 & 20 & 4 & No & 0 & 30 &    0.00\% \\ 
$(\langle 3,3\rangle,13,28)$ & 6 & 84 & 1,4,5 & Yes & 0 & 587 &    0.00\% \\ 
$(\langle 3,3\rangle,18,18)$ & 2 & 116 & 2,4,6 & Yes & 0 & 681 &    0.00\% \\ 
$(\langle 2,2\rangle,3,3,8)$ & 6 & 48 & 2,4,6 & Yes & 0 & 332 &    3.61\% \\ 
$(\langle 2,2\rangle,4,4,4)$ & 6 & 54 & 5 & Yes & 0 & 75 &    0.00\% \\ 
$(\langle 4,4\rangle,5,40)$ & 6 & 48 & 3,5,6 & Yes & 0 & 98 &    0.00\% \\ 
$(\langle 4,4\rangle,6,22)$ & 6 & 60 & 3,5,6 & Yes & 0 & 440 &    0.00\% \\ 
$(\langle 4,4\rangle,7,16)$ & 6 & 72 & 3,5,6 & Yes & 0 & 271 &    0.00\% \\ 
$(\langle 4,4\rangle,8,13)$ & 6 & 48 & 0,1,2,3,4,5,6 & Yes & 0 & 180 &    0.00\% \\ 
$(\langle 3,3\rangle,9,108)$ & 2 & 28 & 4 & No & 0 & 30 &    0.00\% \\ 
$(\langle 6,6\rangle,\langle 6,6\rangle)$ & 4 & 80 & $SO(10)$ only & No & 0 & 152 &    0.00\% \\ 
$(\langle 2,2\rangle,\langle 4,4\rangle,4)$ & 6 & 48 & 5 & Yes & 0 & 103 &    0.00\% \\ 
$(\langle 2,2\rangle,\langle 3,3\rangle,8)$ & 6 & 36 & 4 & No & 0 & 37 &    0.00\% \\ 
$(\langle 2,2\rangle,\langle 2,2\rangle,\langle 2,2\rangle)$ & 6 & 90 & $SO(10)$ only  & No & 0 & 224 &    1.34\% \\ 
$(1,\langle 2,2\rangle,\langle 10,10\rangle)$ & 12 & 60 & 3,5,6 & Yes & 0 & 155 &    0.00\% \\ 
$(\langle 4,4\rangle,\langle 10,10\rangle)$ & 6 & 72 & 5 & Yes & 0 & 142 &    0.00\% \\ 
$(1,\langle 4,4\rangle,\langle 4,4\rangle)$ & 6 & 84 & 3,5,6 & Yes & 0 & 848 &    0.83\% \\ 
\end{longtable}
\end{center}

\LTcapwidth=14truecm
\begin{center}
\vskip .7truecm
\begin{longtable}{|c||c|c|c|c|c|c|c|}\caption{{\bf{Results for lifted Gepner models}}}\\
\hline
 \multicolumn{1}{|c||}{model}
& \multicolumn{1}{c|}{$\Delta$}
& \multicolumn{1}{l|}{Max. }
& \multicolumn{1}{c|}{Groups }
& \multicolumn{1}{l|}{Exotics }
& \multicolumn{1}{c|}{3 family}
& \multicolumn{1}{c|}{$N$ fam.}  
& \multicolumn{1}{c|}{Missing} \\ 
\hline
\endfirsthead
\multicolumn{8}{c}%
{{\bfseries \tablename\ \thetable{} {\rm-- continued from previous page}}} \\
\hline 
 \multicolumn{1}{|c||}{model}
& \multicolumn{1}{c|}{$\Delta$}
& \multicolumn{1}{l|}{Max. }
& \multicolumn{1}{c|}{Groups }
& \multicolumn{1}{l|}{Exotics}
& \multicolumn{1}{c|}{3 family}
& \multicolumn{1}{c|}{$N$ fam.}  
& \multicolumn{1}{c|}{Missing} \\ 
\hline
\endhead
\hline \multicolumn{8}{|r|}{{Continued on next page}} \\ \hline
\endfoot
\hline \hline
\endlastfoot\hline
\label{HWLSummary}
$(\widehat{1},1,1,1,1,\langle 4,4\rangle)$ & 3 & 33 & 3,5,6 & Yes & 45 & 205 &   16.10\% \\ 
$(\widehat{1},1,1,1,\langle 10,10\rangle)$ & 3 & 24 & 3,5,6 & Yes & 0 & 39 &    2.56\% \\ 
$(\widehat{1},1,1,1,\langle 2,2\rangle,4)$ & 3 & 18 & 3,5,6 & Yes & 16 & 50 &   14.00\% \\ 
$(\widehat{1},1,1,\langle 4,4\rangle,4)$ & 3 & 33 & 3,5,6 & Yes & 549 & 1016 &   28.54\% \\ 
$(\widehat{1},1,2,2,\langle 4,4\rangle)$ & 3 & 12 & 3,5,6 & Yes & 17 & 60 &    0.00\% \\ 
$(\widehat{1},1,\langle 2,2\rangle,2,10)$ & 3 & 24 & 3,5,6 & Yes & 123 & 283 &    7.42\% \\ 
$(\widehat{1},1,4,\langle 10,10\rangle)$ & 3 & 24 & 3,5,6 & Yes & 33 & 206 &    7.77\% \\ 
$(\widehat{1},1,\langle 6,6\rangle,10)$ & 6 & 6 & 3,5,6 & Yes & 0 & 15 &    0.00\% \\ 
$(\widehat{1},1,\langle 2,2\rangle,4,4)$ & 3 & 24 & 3,5,6 & Yes & 34 & 237 &    4.64\% \\ 
$(\widehat{1},1,\langle 2,2\rangle,\langle 4,4\rangle)$ & 6 & 12 & 3,5,6 & Yes & 0 & 38 &    0.00\% \\ 
$(\widehat{1},2,2,\langle 10,10\rangle)$ & 12 & 24 & 3,5,6 & Yes & 0 & 18 &    0.00\% \\ 
$(\widehat{1},\langle 2,2\rangle,2,2,4)$ & 6 & 24 & 3,5,6 & Yes & 0 & 71 &    7.04\% \\ 
$(\widehat{1},2,\langle 3,3\rangle,58)$ & 1 & 8 & 0,1,2,3,4,5,6 & Yes & 2 & 40 &    0.00\% \\ 
$(\widehat{1},\langle 2,2\rangle,10,10)$ & 6 & 24 & 3,5,6 & Yes & 0 & 105 &    5.71\% \\ 
$(\widehat{1},\langle 3,3\rangle,4,8)$ & 2 & 8 & 0,1,2,3,4,5,6 & Yes & 0 & 23 &    0.00\% \\ 
$(\widehat{1},\langle 2,2\rangle,6,22)$ & 3 & 24 & 3,5,6 & Yes & 58 & 281 &    5.69\% \\ 
$(\widehat{1},\langle 2,2\rangle,7,16)$ & 6 & 12 & 3,5,6 & Yes & 0 & 13 &    0.00\% \\ 
$(\widehat{2},\langle 2,2\rangle,2,2,2)$ & 1 & 36 & 6 & Yes & 587 & 10481 &    6.91\% \\ 
$(\widehat{2},2,2,\langle 6,6\rangle)$ & 2 & 36 & 6 & Yes & 0 & 595 &    0.17\% \\ 
$(\widehat{2},2,\langle 3,3\rangle,8)$ & 1 & 10 & 2,4,6 & Yes & 3 & 51 &    0.00\% \\ 
$(\widehat{2},2,\langle 2,2\rangle,\langle 2,2\rangle)$ & 2 & 24 & 6 & Yes & 0 & 807 &    1.73\% \\ 
$(\widehat{2},\langle 10,10\rangle,10)$ & 4 & 8 & 3,5,6 & Yes & 0 & 24 &    0.00\% \\ 
$(\widehat{2},\langle 8,8\rangle,18)$ & 1 & 12 & 2,4,6 & Yes & 6 & 85 &    0.00\% \\ 
$(\widehat{2},\langle 2,2\rangle,3,18)$ & 1 & 24 & 2,4,6 & Yes & 26 & 225 &    4.00\% \\ 
$(\widehat{2},\langle 2,2\rangle,4,10)$ & 2 & 24 & 3,5,6 & Yes & 0 & 89 &    3.37\% \\ 
$(\widehat{2},\langle 2,2\rangle,6,6)$ & 1 & 24 & 6 & Yes & 9 & 305 &    1.97\% \\ 
$(\widehat{3},\langle 6,6\rangle,18)$ & 2 & 20 & 2,4,6 & No & 0 & 85 &    0.00\% \\ 
$(\widehat{3},\langle 5,5\rangle,68)$ & 12 & 12 & 4 & No & 0 & 4 &    0.00\% \\ 
$(\widehat{3},\langle 8,8\rangle,8)$ & 1 & 48 & 2,4,6 & Yes & 146 & 1709 &    0.06\% \\ 
$(\widehat{3},\langle 3,3\rangle,\langle 3,3\rangle)$ & 1 & 28 & 4 & No & 11 & 80 &    0.00\% \\ 
$(\widehat{4},4,\langle 10,10\rangle)$ & 2 & 16 & 5 & Yes & 0 & 47 &    0.00\% \\ 
$(\widehat{4},\langle 7,7\rangle,7)$ & 1 & 3 & 5 & Yes & 2 & 6 &    0.00\% \\ 
$(\widehat{6},\langle 6,6\rangle,6)$ & 1 & 8 & 6 & Yes & 7 & 85 &    0.00\% \\ 
$(\langle 3,3\rangle,10,\widehat{58})$ & 1 & 6 & 0,1,2,3,4,5,6 & Yes & 2 & 11 &    0.00\% \\ 
$(\langle 3,3\rangle,\widehat{12},33)$ & 2 & 6 & 4 & No & 0 & 5 &    0.00\% \\ 
$(\langle 3,3\rangle,\widehat{13},28)$ & 1 & 30 & 1,4,5 & Yes & 29 & 170 &    0.00\% \\ 
$(\langle 2,2\rangle,\widehat{3},3,8)$ & 1 & 24 & 2,4,6 & Yes & 14 & 479 &    4.38\% \\ 
$(\langle 2,2\rangle,\widehat{4},4,4)$ & 2 & 24 & 5 & Yes & 0 & 111 &    0.90\% \\ 
$(\langle 4,4\rangle,\widehat{6},22)$ & 8 & 8 & 3,5,6 & Yes & 0 & 5 &    0.00\% \\ 
$(\langle 4,4\rangle,\widehat{8},13)$ & 4 & 16 & 0,1,2,3,4,5,6 & Yes & 0 & 17 &    0.00\% \\ 
$(\langle 3,3\rangle,\widehat{9},108)$ & 2 & 6 & 4 & No & 0 & 6 &    0.00\% \\ 
$(\widehat{4},\langle 2,2\rangle,\langle 4,4\rangle)$ & 4 & 20 & 5 & Yes & 0 & 53 &    0.00\% \\ 
$(\langle 2,2\rangle,\langle 3,3\rangle,\widehat{8})$ & 2 & 6 & 4 & No & 0 & 11 &    0.00\% \\ 
$(1,1,1,\widehat{4},\langle 4,4\rangle)$ & 1 & 24 & 3,5,6 & Yes & 78 & 859 &   20.84\% \\ 
$(1,1,\widehat{2},2,\langle 4,4\rangle)$ & 1 & 6 & 3,5,6 & Yes & 0 & 75 &    4.00\% \\ 
$(1,1,\widehat{2},\langle 2,2\rangle,10)$ & 1 & 24 & 3,5,6 & Yes & 20 & 323 &    8.05\% \\ 
$(1,1,\widehat{4},\langle 10,10\rangle)$ & 2 & 16 & 3,5,6 & Yes & 0 & 191 &    4.71\% \\ 
$(1,1,\langle 2,2\rangle,\widehat{4},4)$ & 2 & 32 & 3,5,6 & Yes & 0 & 297 &    5.05\% \\ 
$(1,\widehat{2},2,\langle 10,10\rangle)$ & 1 & 16 & 3,5,6 & Yes & 28 & 262 &    0.00\% \\ 
$(1,\langle 2,2\rangle,2,2,\widehat{4})$ & 4 & 32 & 3,5,6 & Yes & 0 & 118 &    8.47\% \\ 
$(1,\widehat{2},4,\langle 6,6\rangle)$ & 2 & 24 & 3,5,6 & Yes & 0 & 77 &    0.00\% \\ 
$(1,\widehat{2},\langle 4,4\rangle,10)$ & 1 & 16 & 3,5,6 & Yes & 147 & 1160 &    8.71\% \\ 
$(1,\widehat{2},\langle 3,3\rangle,58)$ & 1 & 8 & 0,1,2,3,4,5,6 & Yes & 3 & 56 &    0.00\% \\ 
$(1,\widehat{4},\langle 4,4\rangle,4)$ & 1 & 24 & 3,5,6 & Yes & 425 & 3219 &    6.28\% \\ 
$(1,\langle 3,3\rangle,\widehat{4},8)$ & 2 & 10 & 0,1,2,3,4,5,6 & Yes & 0 & 27 &    0.00\% \\ 
$(1,\langle 2,2\rangle,\widehat{6},22)$ & 1 & 48 & 3,5,6 & Yes & 46 & 645 &    3.41\% \\ 
$(2,2,\widehat{4},\langle 4,4\rangle)$ & 4 & 20 & 3,5,6 & Yes & 0 & 93 &    0.00\% \\ 
$(2,2,\langle 3,3\rangle,\widehat{8})$ & 2 & 8 & 2,4,6 & Yes & 0 & 29 &    0.00\% \\ 
$(\langle 2,2\rangle,2,\widehat{3},18)$ & 2 & 30 & 2,4,6 & Yes & 0 & 380 &    2.63\% \\ 
$(\langle 2,2\rangle,2,\widehat{4},10)$ & 4 & 24 & 3,5,6 & Yes & 0 & 107 &    0.93\% \\ 
$(\langle 2,2\rangle,2,\widehat{6},6)$ & 2 & 48 & 6 & Yes & 0 & 477 &    2.73\% \\ 
$(3,\widehat{8},\langle 8,8\rangle)$ & 1 & 32 & 2,4,6 & Yes & 24 & 480 &    0.00\% \\ 
$(\langle 5,5\rangle,5,\widehat{12})$ & 1 & 6 & $SO(10)$ only & No & 2 & 8 &    0.00\% \\ 
$(\langle 2,2\rangle,3,3,\widehat{8})$ & 1 & 18 & 2,4,6 & Yes & 6 & 116 &    0.00\% \\ 
$(\langle 4,4\rangle,8,\widehat{13})$ & 2 & 14 & 0,1,2,3,4,5,6 & Yes & 0 & 20 &    0.00\% \\ 
$(1,\widehat{2},\langle 2,2\rangle,2,4)$ & 2 & 24 & 3,5,6 & Yes & 0 & 1092 &    4.85\% \\ 
$(1,2,\langle 3,3\rangle,\widehat{58})$ & 2 & 12 & 0,1,2,3,4,5,6 & Yes & 0 & 11 &    0.00\% \\ 
$(1,\langle 3,3\rangle,4,\widehat{8})$ & 2 & 6 & 0,1,2,3,4,5,6 & Yes & 0 & 10 &    0.00\% \\ 
\end{longtable}
\end{center}

\LTcapwidth=14truecm
\begin{center}
\vskip .7truecm
\begin{longtable}{|c||c|c|c|c|c|c|c|}\caption{{\bf{Results for B-L lifted (lift A) Gepner models}}}\\
\hline
 \multicolumn{1}{|c||}{model}
& \multicolumn{1}{c|}{$\Delta$}
& \multicolumn{1}{l|}{Max. }
& \multicolumn{1}{c|}{Groups }
& \multicolumn{1}{l|}{Exotics }
& \multicolumn{1}{c|}{3 family}
& \multicolumn{1}{c|}{$N$ fam.}  
& \multicolumn{1}{c|}{Missing} \\ 
\hline
\endfirsthead
\multicolumn{8}{c}%
{{\bfseries \tablename\ \thetable{} {\rm-- continued from previous page}}} \\
\hline 
 \multicolumn{1}{|c||}{model}
& \multicolumn{1}{c|}{$\Delta$}
& \multicolumn{1}{l|}{Max. }
& \multicolumn{1}{c|}{Groups }
& \multicolumn{1}{l|}{Exotics}
& \multicolumn{1}{c|}{3 family}
& \multicolumn{1}{c|}{$N$ fam.}  
& \multicolumn{1}{c|}{Missing} \\ 
\hline
\endhead
\hline \multicolumn{8}{|r|}{{Continued on next page}} \\ \hline
\endfoot
\hline \hline
\endlastfoot\hline
\label{BLASummary}
$(1,2,\langle 3,3\rangle,58)$ & 1 & 6 & 0,1,2,3,4,5,6 & Yes & 4 & 33 &    0.00\% \\ 
$(1,\langle 3,3\rangle,4,8)$ & 2 & 6 & 0,1,2,3,4,5,6 & Yes & 0 & 12 &    0.00\% \\ 
$(2,2,\langle 3,3\rangle,8)$ & 2 & 8 & 2,4,6 & Yes & 0 & 25 &    0.00\% \\ 
$(2,\langle 8,8\rangle,18)$ & 1 & 12 & 2,4,6 & Yes & 14 & 90 &    0.00\% \\ 
$(\langle 2,2\rangle,2,3,18)$ & 2 & 18 & 2,4,6 & Yes & 0 & 364 &    3.85\% \\ 
$(3,\langle 6,6\rangle,18)$ & 2 & 14 & 2,4,6 & No & 0 & 84 &    0.00\% \\ 
$(3,\langle 5,5\rangle,68)$ & 6 & 6 & 4 & No & 0 & 12 &    0.00\% \\ 
$(3,\langle 8,8\rangle,8)$ & 1 & 30 & 2,4,6 & Yes & 238 & 1799 &    0.11\% \\ 
$(3,\langle 3,3\rangle,\langle 3,3\rangle)$ & 1 & 18 & 4 & No & 15 & 84 &    0.00\% \\ 
$(\langle 3,3\rangle,10,58)$ & 1 & 10 & 0,1,2,3,4,5,6 & Yes & 1 & 19 &    0.00\% \\ 
$(\langle 3,3\rangle,12,33)$ & 2 & 4 & 4 & No & 0 & 6 &    0.00\% \\ 
$(\langle 3,3\rangle,13,28)$ & 1 & 9 & 1,4,5 & Yes & 57 & 346 &    0.00\% \\ 
$(\langle 3,3\rangle,18,18)$ & 1 & 14 & 2,4,6 & Yes & 30 & 200 &    0.00\% \\ 
$(\langle 2,2\rangle,3,3,8)$ & 1 & 12 & 2,4,6 & Yes & 30 & 246 &    0.00\% \\ 
$(\langle 4,4\rangle,8,13)$ & 2 & 8 & 0,1,2,3,4,5,6 & Yes & 0 & 49 &    0.00\% \\ 
$(\langle 3,3\rangle,9,108)$ & 2 & 4 & 4 & No & 0 & 6 &    0.00\% \\ 
$(\langle 2,2\rangle,\langle 3,3\rangle,8)$ & 2 & 6 & 4 & No & 0 & 12 &    0.00\% \\ 
\end{longtable}
\end{center}




\LTcapwidth=14truecm
\begin{center}
\vskip .7truecm
\begin{longtable}{|c||c|c|c|c|c|c|c|}\caption{{\bf{Results for B-L lifted Gepner (lift B) models}}}\\
\hline
 \multicolumn{1}{|c||}{model}
& \multicolumn{1}{c|}{$\Delta$}
& \multicolumn{1}{l|}{Max. }
& \multicolumn{1}{c|}{Groups }
& \multicolumn{1}{l|}{Exotics }
& \multicolumn{1}{c|}{3 family}
& \multicolumn{1}{c|}{$N$ fam.}  
& \multicolumn{1}{c|}{Missing} \\ 
\hline
\endfirsthead
\multicolumn{8}{c}%
{{\bfseries \tablename\ \thetable{} {\rm-- continued from previous page}}} \\
\hline 
 \multicolumn{1}{|c||}{model}
& \multicolumn{1}{c|}{$\Delta$}
& \multicolumn{1}{l|}{Max. }
& \multicolumn{1}{c|}{Groups }
& \multicolumn{1}{l|}{Exotics}
& \multicolumn{1}{c|}{3 family}
& \multicolumn{1}{c|}{$N$ fam.}  
& \multicolumn{1}{c|}{Missing} \\ 
\hline
\endhead
\hline \multicolumn{8}{|r|}{{Continued on next page}} \\ \hline
\endfoot
\hline \hline
\endlastfoot\hline
\label{BLBSummary}
$(1,2,\langle 3,3\rangle,58)$ & 2 & 10 & 0,1,2,3,4,5,6 & Yes & 0 & 32 &    0.00\% \\ 
$(1,\langle 3,3\rangle,4,8)$ & 2 & 10 & 0,1,2,3,4,5,6 & Yes & 0 & 10 &    0.00\% \\ 
$(2,2,\langle 3,3\rangle,8)$ & 2 & 14 & 2,4,6 & Yes & 0 & 34 &    0.00\% \\ 
$(2,\langle 8,8\rangle,18)$ & 2 & 16 & 2,4,6 & Yes & 0 & 108 &    0.00\% \\ 
$(\langle 2,2\rangle,2,3,18)$ & 2 & 36 & 2,4,6 & Yes & 0 & 476 &    3.99\% \\ 
$(3,\langle 6,6\rangle,18)$ & 4 & 28 & 2,4,6 & No & 0 & 82 &    0.00\% \\ 
$(3,\langle 5,5\rangle,68)$ & 8 & 16 & 4 & No & 0 & 12 &    0.00\% \\ 
$(3,\langle 8,8\rangle,8)$ & 2 & 56 & 2,4,6 & Yes & 0 & 1781 &    0.00\% \\ 
$(3,\langle 3,3\rangle,\langle 3,3\rangle)$ & 2 & 32 & 4 & No & 0 & 81 &    0.00\% \\ 
$(\langle 3,3\rangle,10,58)$ & 2 & 18 & 0,1,2,3,4,5,6 & Yes & 0 & 18 &    0.00\% \\ 
$(\langle 3,3\rangle,12,33)$ & 2 & 8 & 4 & No & 0 & 6 &    0.00\% \\ 
$(\langle 3,3\rangle,13,28)$ & 2 & 18 & 1,4,5 & Yes & 0 & 322 &    0.00\% \\ 
$(\langle 3,3\rangle,18,18)$ & 2 & 26 & 2,4,6 & Yes & 0 & 191 &    0.00\% \\ 
$(\langle 2,2\rangle,3,3,8)$ & 2 & 24 & 2,4,6 & Yes & 0 & 226 &    0.00\% \\ 
$(\langle 4,4\rangle,8,13)$ & 4 & 16 & 0,1,2,3,4,5,6 & Yes & 0 & 45 &    0.00\% \\ 
$(\langle 3,3\rangle,9,108)$ & 2 & 10 & 4 & No & 0 & 6 &    0.00\% \\ 
$(\langle 2,2\rangle,\langle 3,3\rangle,8)$ & 2 & 10 & 4 & No & 0 & 10 &    0.00\% \\ 
\end{longtable}
\end{center}

\begin{figure}[p]
\begin{center}
\includegraphics[width=17cm]{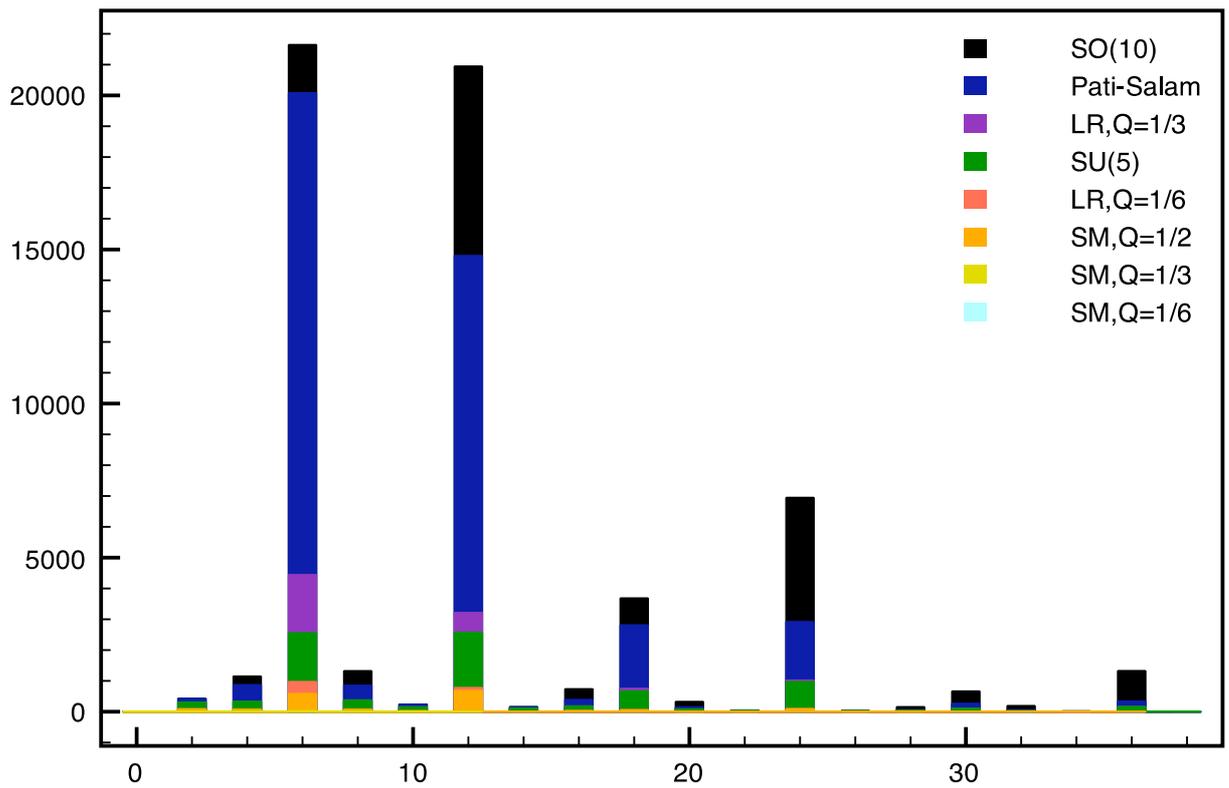}
\caption{\small Distribution of the number of families for permutation orbifolds of standard Gepner Models.}
\label{famplot_standard}
\end{center}
\end{figure}

\begin{figure}[p]
\begin{center}
\includegraphics[width=17cm]{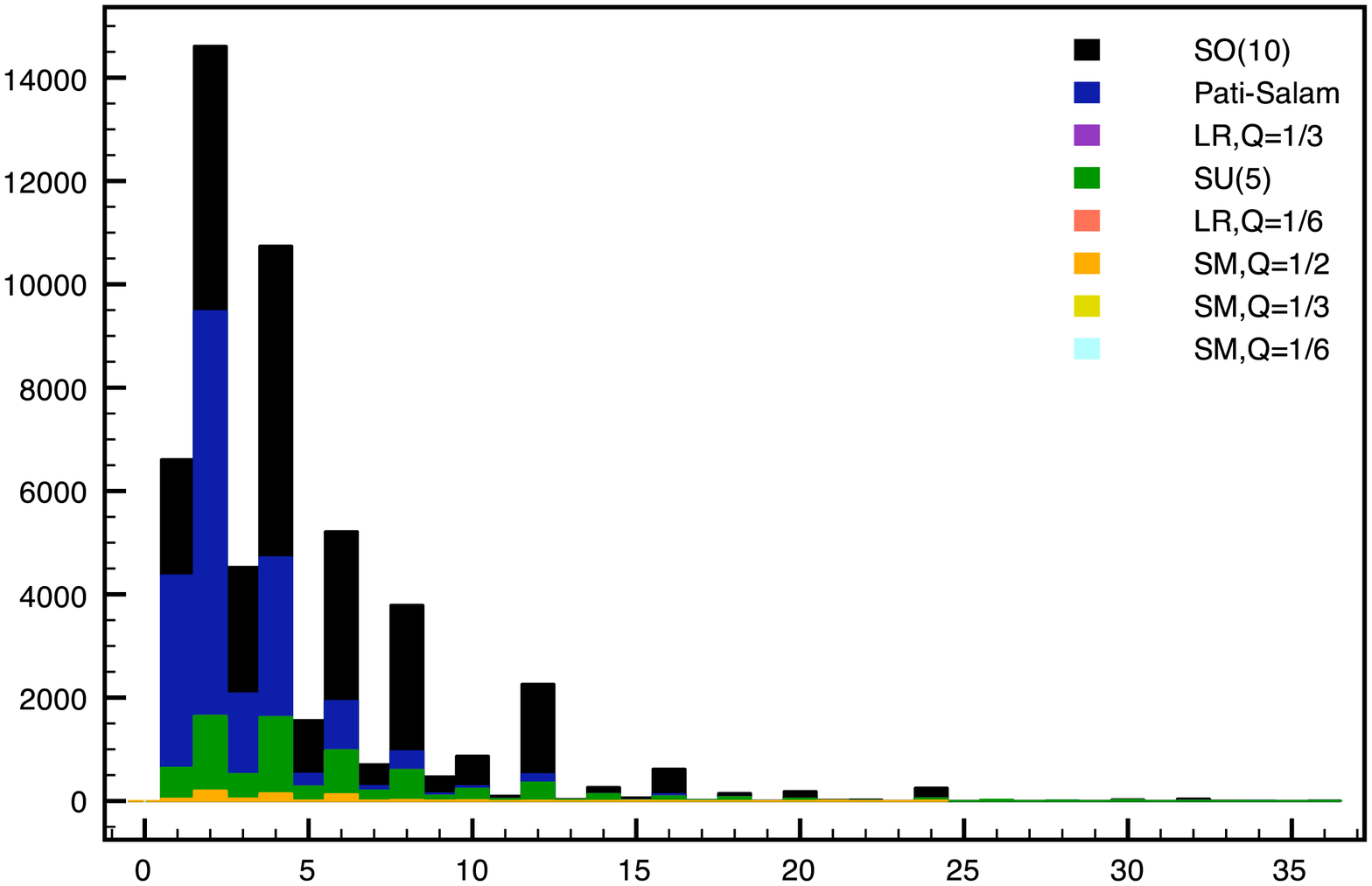}
\caption{\small Distribution of the number of families for permutation orbifolds of lifted Gepner Models.}
\label{famplot_lift}
\end{center}
\end{figure}

\begin{figure}[p]
\begin{center}
\includegraphics[width=17cm]{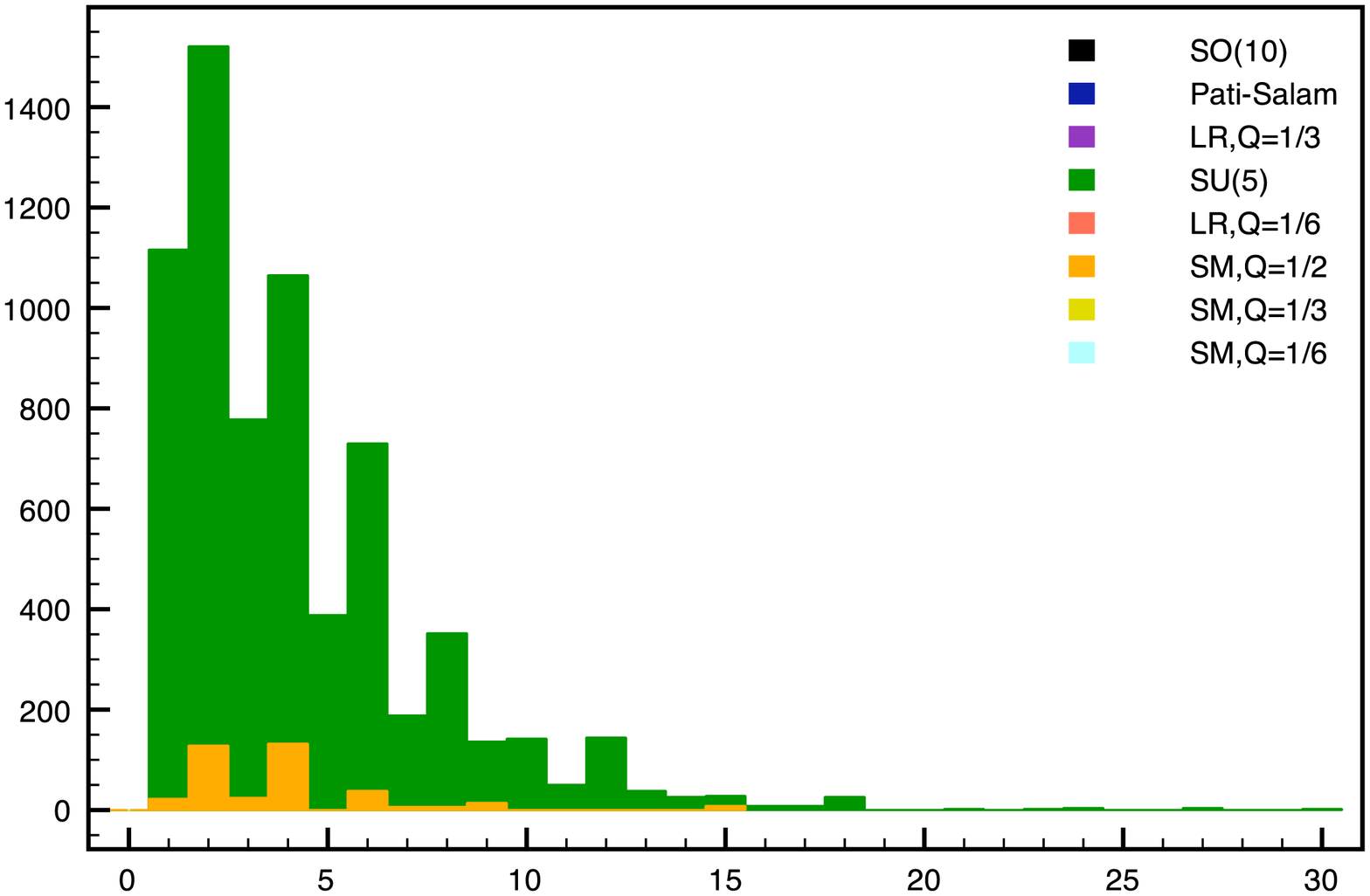}
\caption{\small Distribution of the number of families for permutation orbifolds of B-L lifted (lift A) Gepner Models.}
\label{famplot_liftA}
\end{center}
\end{figure}

\begin{figure}[p]
\begin{center}
\includegraphics[width=17cm]{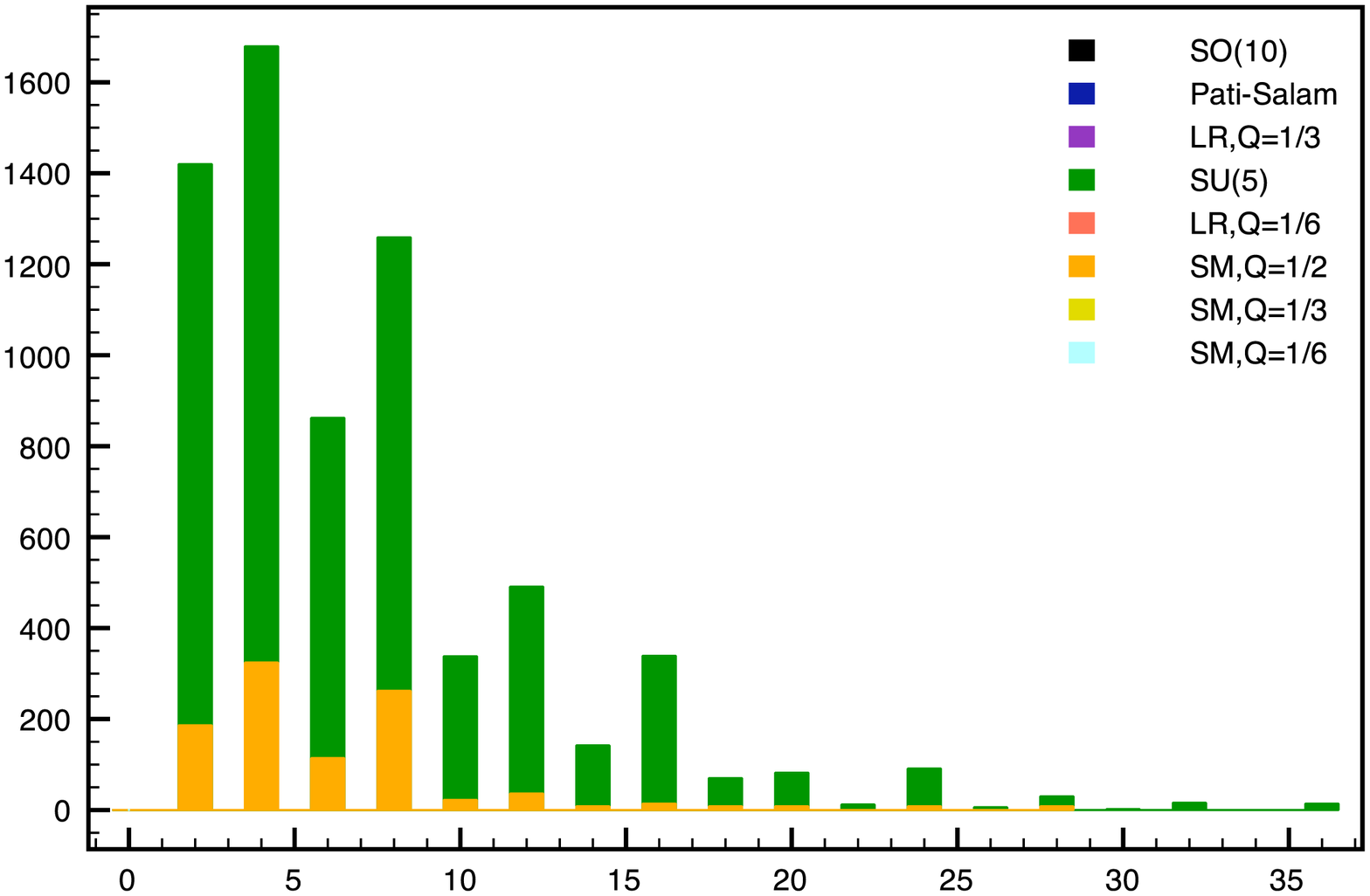}
\caption{\small Distribution of the number of families for permutation orbifolds of B-L lifted (lift B) Gepner Models.}
\label{famplot_liftB}
\end{center}
\end{figure}

\end{document}